\PassOptionsToPackage{hyphens}{url}
\PassOptionsToPackage{table}{xcolor}
\documentclass{vldb}

\usepackage{amsmath,amssymb,balance}
\usepackage[inline]{enumitem}
\usepackage[final,hidelinks]{hyperref}
\usepackage[final]{microtype}
\usepackage[binary-units]{siunitx}
\DeclareSIUnit{\nothing}{\relax}
\usepackage{amsthm}
\newtheorem{theorem}{Theorem}[section]
\newtheorem{proposition}[theorem]{Proposition}
\newtheorem{lemma}[theorem]{Lemma}
\theoremstyle{definition}
\newtheorem{definition}[theorem]{Definition}
\theoremstyle{remark}
\newtheorem{example}[theorem]{Example}
\newtheorem{remark}[theorem]{Remark}

\newcommand{\Replicas}{\mathfrak{R}}
\newcommand{\Clusters}{\mathfrak{S}}
\newcommand{\Cluster}{\mathcal{C}}
\newcommand{\Clients}[1]{\operatorname{clients}(#1)}
\newcommand{\Replica}[1][r]{\MakeUppercase{#1}}
\newcommand{\ID}[1]{\operatorname{id}(#1)}
\newcommand{\PrimaryC}[1]{\mathsf{P}_{#1}}
\newcommand{\Client}{c}
\newcommand{\Faulty}[1]{\mathsf{f}(#1)}
\newcommand{\NonFaulty}[1]{\mathsf{nf}(#1)}
\newcommand{\n}{\mathbf{n}}
\newcommand{\f}{\mathbf{f}}
\newcommand{\z}{\mathbf{z}}
\newcommand{\rn}{\rho}

\newcommand{\Transaction}{T}
\newcommand{\Message}[2]{\textsc{#1}(#2)}
\newcommand{\SignMessage}[2]{\langle#2\rangle_{#1}}
\newcommand{\Cert}[2]{[#2]_{#1}}

\newcommand{\Name}[1]{\textnormal{\textsc{#1}}}
\newcommand{\GeoBFT}{\Name{GeoBFT}}
\newcommand{\BFT}{\Name{Bft}}
\newcommand{\PBFT}{\Name{Pbft}}
\newcommand{\ZZ}{\Name{Zyzzyva}}
\newcommand{\HS}{\Name{HotStuff}}
\newcommand{\STW}{\Name{Steward}}
\newcommand{\ResilientDB}{\Name{Resilient\-DB}}
\newcommand{\PoW}{\Name{PoW}}

\newcommand{\abs}[1]{\lvert #1 \rvert}

\newcommand{\BigO}[1]{\mathcal{O}(#1)}
\newcommand{\dsfloorfrac}[2]{\lfloor #1 / #2 \rfloor}

\newcommand{\intersect}{\cap}
\newcommand{\difference}{\setminus}

\usepackage{fontawesome}

\usepackage{tikz,pgfplots,pgfplotstable}
\usetikzlibrary{arrows.meta,decorations.pathmorphing,decorations.pathreplacing,patterns}
\tikzset{
    >=Stealth,
    dot/.style={circle,scale=0.35,draw=black,fill=black},
    label/.append style={font=\strut\scriptsize},
    geobftplot/.append style={baseline,scale=0.5425}
}

\pgfplotscreateplotcyclelist{geobft}{
    black,mark=*\\
    blue,mark=square*\\
    red,mark=otimes*\\
    orange,mark=triangle*\\
    brown,mark=diamond*\\
}

\pgfplotsset{
    tick label style={font=\large},
    label style={font=\large},
    legend style={font=\large},
    every axis/.append style={
        ylabel near ticks,
        mark size=1.5pt,
        cycle list name=geobft,
        ymin=0,
        enlargelimits,
        font=\large
    }
}

\usepackage[noend]{algorithmic}
\newenvironment{myprotocol}{
    \hrule
    \smallskip
    \scriptsize
    \algsetup{linenosize=\tiny}
    \begin{algorithmic}[1]
        \newcommand{\SPACE}{\item[]}
        \newcommand{\GETS}{:=}
        \newcommand{\TITLE}[2]{\item[] \textbf{\underline{##1}} (##2) \textbf{:}\\[0.5pt]}
        \makeatletter
            \newcommand{\EVENT}[1]{\STATE \textbf{event} ##1 \textbf{do}\begin{ALC@g}}
            \newcommand{\ENDEVENT}{\end{ALC@g}}
        \makeatother
}{
    \end{algorithmic}
    \smallskip
    \hrule
}

\begin{document}

\title{ResilientDB: Global Scale Resilient Blockchain Fabric}
\numberofauthors{1}
\author{
\alignauthor
    Suyash Gupta\thanks{Both authors have equally contributed to this work.}  \qquad 
    Sajjad Rahnama\footnotemark[1] \qquad
    Jelle Hellings \qquad
    Mohammad Sadoghi \\[4pt]	
        \affaddr{Exploratory Systems Lab}\\
	\affaddr{ Department of Computer Science}\\
	\affaddr{University of California, Davis}\\
        \email{\textbraceleft sgupta,srahnama,jhellings,msadoghi\textbraceright @ucdavis.edu}\\
}

\vldbTitle{ResilientDB: Global Scale Resilient Blockchain Fabric}
\vldbAuthors{Suyash Gupta, Sajjad Rahnama, Jelle Hellings, and Mohammad Sadoghi}
\vldbDOI{https://doi.org/10.14778/3380750.3380757}
\vldbVolume{13}
\vldbNumber{6}
\vldbYear{2020}

\maketitle

\begin{abstract}
Recent developments in blockchain technology have inspired innovative new designs in resilient distributed and database systems. At their core, these blockchain applications typically use Byzantine fault-tolerant consensus protocols to maintain a common state across all replicas, even if some replicas are faulty or malicious. Unfortunately, existing consensus protocols are not designed to deal with \emph{geo-scale deployments} in which many replicas spread across a geographically large area participate in consensus.
 
To address this, we present the Geo-Scale Byzantine Fault-Tolerant consensus protocol (\GeoBFT). \GeoBFT{} is designed for excellent scalability by using a topological-aware grouping of replicas in local clusters, by introducing parallelization of consensus at the local level, and by minimizing communication between clusters. To validate our vision of high-performance geo-scale resilient distributed systems, we implement \GeoBFT{} in our efficient \ResilientDB{} permissioned blockchain fabric. We show that \GeoBFT{} is not only sound and  provides great scalability, but also outperforms state-of-the-art consensus protocols by a factor of six in geo-scale deployments.

\end{abstract}

\section{Introduction}\label{sec:intro}
Recent interest in \emph{blockchain technology} has renewed development of distributed \emph{Byzantine fault-tolerant} (\BFT{}) systems that can deal with failures and malicious attacks of some 
participants~\cite{blockdev,blockeu,impactblock,christies,ibmgdpr,disc_mbft,encybd,blockhealthover,bitcoin,blockplane,hypereal,ether}.
Although these systems are safe, they attain low throughput, especially  when the nodes are spread across a wide-area network 
(or \emph{geographically large distances}).
We believe this contradicts the central promises of blockchain technology:
\emph{decentralization} and \emph{democracy}, 
in which arbitrary replicas at arbitrary distances can participate~\cite{blockbench,tutorial-middleware,encybd}.

At the core of any blockchain system is a \BFT{} consensus protocol that helps participating replicas to achieve resilience. Existing blockchain database systems and data-processing frameworks typically use \emph{permissioned blockchain designs} that rely on traditional \BFT{} consensus~\cite{poe,multibft-system,blockchain_dist,distbook,pdbook,distalgo}. These permissioned blockchains employ a \emph{fully-replicated design} in which all replicas are known and each replica holds a full copy of the data (the blockchain).

\subsection{Challenges for Geo-scale Blockchains}\label{ss:challenges}
To enable  geo-scale deployment of a permissioned blockchain system, we believe that the underlying consensus protocol must distinguish 
between \emph{local} and \emph{global} communication. 
This belief is easily supported in practice. 
For example, in Table~\ref{tbl:geoscale_cost} we illustrate the ping round-trip time and bandwidth measurements. 
These measurements show that global message latencies are at least $33$--$270$ times higher than local latencies, 
while the maximum throughput is $10$--$151$ times lower, both implying that communication between regions is 
\emph{several orders of magnitude} more costly than communication within regions. 
Hence, a blockchain system needs to recognize and minimize global communication if it is to attain high performance in a geo-scale deployment.

\begin{table}[t!]
    \newcommand{\LocalCell}[1]{\cellcolor{green!35}#1}
    \newcommand{\GlobalCell}[1]{\cellcolor{red!20}#1}
    \newcommand{\RCell}[0]{\cellcolor{black!10}}
    \caption{Real-world inter- and intra-cluster communication costs in terms of the ping round-trip times (which determines \emph{latency}) and bandwidth (which determines \emph{throughput}). These measurements are taken in Google Cloud using clusters of \texttt{n1} machines (replicas) that are deployed in six different regions.}\label{tbl:geoscale_cost}
    \smallskip
    \centering
    \setlength\tabcolsep{1pt}
    \scalebox{0.78}{
        \begin{tabular}{|l|cccccc|cccccc|}
        \hline
        \multicolumn{1}{|c|}{}&\multicolumn{6}{c|}{\emph{Ping round-trip times (\si{\milli\second})}}&\multicolumn{6}{c|}{\emph{Bandwidth (\si[per-mode=symbol]{\mega\bit\per\second})}}\\
        \hline
        \multicolumn{1}{|c|}{}& \emph{O}& \emph{I}& \emph{M}& \emph{B}& \emph{T}& \emph{S}& \emph{O}& \emph{I}& \emph{M}& \emph{B}& \emph{T}& \emph{S}\\
        \hline
        \hline
        Oregon (\emph{O})&\LocalCell{$\leq1$}&\GlobalCell{$38$}&\GlobalCell{$65$}&\GlobalCell{$136$}&\GlobalCell{$118$}&\GlobalCell{$161$}&\LocalCell{$7998$}&\GlobalCell{$669$}&\GlobalCell{$371$}&\GlobalCell{$194$}&\GlobalCell{$188$}&\GlobalCell{$136$}\\
        Iowa (\emph{I})&\RCell{}&\LocalCell{$\leq1$}&\GlobalCell{$33$}&\GlobalCell{$98$}&\GlobalCell{$153$}&\GlobalCell{$172$}            &\RCell{}&\LocalCell{$10004$}&\GlobalCell{$752$}&\GlobalCell{$243$}&\GlobalCell{$144$}&\GlobalCell{$120$}\\         
        Montreal (\emph{M})&\RCell{}&\RCell{}&\LocalCell{$\leq1$}&\GlobalCell{$82$}&\GlobalCell{$186$}&\GlobalCell{$202$}                 &\RCell{}&\RCell{}&\LocalCell{$7977$}&\GlobalCell{$283$}&\GlobalCell{$111$}&\GlobalCell{$102$}\\
        Belgium (\emph{B})&\RCell{}&\RCell{}&\RCell{}&\LocalCell{$\leq1$}&\GlobalCell{$252$}&\GlobalCell{$270$}                           &\RCell{}&\RCell{}&\RCell{}&\LocalCell{$9728$}&\GlobalCell{$79$}&\GlobalCell{$66$}\\
        Taiwan (\emph{T})&\RCell{}&\RCell{}&\RCell{}&\RCell{}&\LocalCell{$\leq1$}&\GlobalCell{$137$}                                      &\RCell{}&\RCell{}&\RCell{}&\RCell{}&\LocalCell{$7998$}&\GlobalCell{$160$}\\
        Sydney (\emph{S})&\RCell{}&\RCell{}&\RCell{}&\RCell{}&\RCell{}&\LocalCell{$\leq1$}                                                &\RCell{}&\RCell{}&\RCell{}&\RCell{}&\RCell{}&\LocalCell{$7977$}\\
        \hline
        \end{tabular}
    }
\end{table}

In the design of geo-scale aware consensus protocols, this translates to two important properties. 
First, a geo-scale aware consensus protocol needs to be \emph{aware of the network topology}. 
This can be achieved by clustering replicas in a region together and favoring communication within 
such clusters over global inter-cluster communication. Second, a geo-scale aware consensus protocol needs to be \emph{decentralized}: no single replica or cluster should be responsible for coordinating all consensus decisions, as such a centralized design limits the throughput to the outgoing global bandwidth and latency of this single replica or cluster.

Existing state-of-the-art consensus protocols do not share these two properties. The influential Practical Byzantine Fault Tolerance consensus protocol (\PBFT{})~\cite{pbft,pbftj} is centralized, as it relies on a single primary replica to coordinate all consensus decisions, 
and requires a vast amount of global communication (between all pairs of replicas). 
Protocols such as \ZZ{} improve on this by reducing communication costs in the optimal case~\cite{bft700,zyzzyva,zyzzyvaj}. 
However, these protocols still have a highly centralized design and do not favor local communication. 
Furthermore, \ZZ{} provides high throughput only if there are no failures and requires reliable clients~\cite{zzuns,aadvark}. 
The recently introduced \HS{} improves on \PBFT{} by simplifying the recovery process on primary failure~\cite{hotstuff}. 
This allows \HS{} to efficiently switch primaries for every consensus decision, providing the potential of decentralization. 
However, the design of \HS{} does not favor local communication, and the usage of threshold signatures strongly centralizes 
all communication for a single consensus decision to the primary of that round. 
Another recent protocol \textsc{PoE} provides better throughput than both 
\PBFT{} and \ZZ{} in the presence of failures, this without employing threshold signatures~\cite{poe}.
Unfortunately, also \textsc{PoE} has a centralized design that depends on a single primary.
Finally, the geo-aware consensus protocol \STW{} promises to do better~\cite{steward}, as it recognizes local clusters and tries to minimize inter-cluster communication. However, due to its centralized design and reliance on cryptographic primitives with high computational costs, \STW{} is unable to benefit from its topological knowledge of the network.

\subsection{GeoBFT: Towards Geo-scale Consensus}
In this work, we improve on the state-of-the-art by introducing \GeoBFT{}, 
a topology-aware and decentralized consensus protocol. 
In \GeoBFT{}, we group replicas in a region into clusters, and we let each cluster make consensus decisions independently. 
These consensus decisions are then shared via an optimistic low-cost communication protocol with the other clusters, 
in this way assuring that all replicas in all clusters are able to learn the same sequence of consensus decisions: if we have two clusters $\Cluster_1$ and $\Cluster_2$ with $\n$ replicas each, then our optimistic communication protocol requires only $\lceil \n/3 \rceil$ messages to be sent from $\Cluster_1$ to $\Cluster_2$ when $\Cluster_1$ needs to share local consensus decisions with $\Cluster_2$.
In specific, we make the following contributions:
\begin{enumerate}
    \item We introduce the \GeoBFT{} consensus protocol, a novel consensus protocol 
that performs a topological-aware grouping of replicas into local clusters to minimize global communication. \GeoBFT{} also decentralizes consensus by allowing each cluster to make consensus decisions independently.
    \item To reduce global communication, we introduce a novel global sharing protocol that 
{\em optimistically} performs minimal inter-cluster communication, while still enabling reliable detection of communication failure.
    \item The optimistic global sharing protocol is supported by a novel \emph{remote view-change protocol} that deals with any malicious behavior and any failures.
    \item We prove that \GeoBFT{} guarantees \emph{safety}: it achieves a unique sequence of consensus decisions among all replicas and ensures that clients can reliably detect when their transactions are executed, this independent of any malicious behavior by any replicas.
    \item We show that \GeoBFT{} guarantees \emph{liveness}: whenever the network provides reliable communication, \GeoBFT{} continues successful operation, this independent of any malicious behavior by any replicas.
    \item To validate our vision of using \GeoBFT{} in geo-scale settings, we present  our \ResilientDB{} fabric~\cite{resilientdb} and implement \GeoBFT{} in this fabric.\footnote{We have open-sourced our \ResilientDB{} fabric at \url{https://resilientdb.com/}.}
    \item We also implemented other state-of-the-art \BFT{} protocols in \ResilientDB{} (\ZZ{}, \PBFT{}, \HS{}, and \STW{}), and evaluate \GeoBFT{} against these \BFT{} protocols using the YCSB benchmark~\cite{ycsb}. We show that \emph{\GeoBFT{} achieves up-to-six times more throughput} than existing \BFT{} protocols.
\end{enumerate}

In Table~\ref{tbl:compare}, we provide a summary of the complexity of the normal-case operations of \GeoBFT{} and compare this to the complexity of other popular \BFT{} protocols.

\begin{table}[t!]
    \newcommand{\Bad}{\cellcolor{red!20}}
    \newcommand{\Good}{\cellcolor{green!35}}
    \caption{The normal-case metrics of \BFT{} consensus protocols in a system with $\z$ clusters, each with $\n$ replicas of which at most $\f$, $\n > 3\f$, are Byzantine. \GeoBFT{} provides the lowest global communication cost per consensus decision (transaction) and operates decentralized.}\label{tbl:compare}
    \smallskip
    \centering
    \scalebox{0.78}{
        \begin{tabular}{|l||c|c|c|c|}
        \hline
        Protocol&Decisions&\multicolumn{2}{c|}{Communication}&Centralized\\
        \hline
                &        & \emph{(Local)} & \emph{(Global)} &\\
        \hline
        \hline
        \Good{}\GeoBFT{} (our paper)  &\Good{}$\z$&\Good{} $\BigO{2\z\n^2}$&\Good{} $\BigO{\f\z^2}$&\Good{}No\\
        \Good{}\qquad\rotatebox[origin=c]{180}{$\Lsh$} \emph{single decision}&\Good{}$1$&\Good{} $\BigO{4\n^2}$&\Good{}$\BigO{\f\z}$&\Good{}No\\
        \STW{}      &$1$     &$\BigO{2\z\n^2}$& \Bad{}$\BigO{\z^2}$&\Bad{}Yes\\
        \hline
        \ZZ{}       &$1$ & \multicolumn{2}{c|}{$\BigO{\z\n}$}      &\Bad{}Yes\\
        \PBFT{}     &$1$ & \multicolumn{2}{c|}{\Bad{}$\BigO{2(\z\n)^2}$} &\Bad{}Yes\\
	\textsc{PoE}&$1$ & \multicolumn{2}{c|}{\Bad{}$\BigO{(\z\n)^2}$} &\Bad{}Yes\\
        \HS{}       &$1$ & \multicolumn{2}{c|}{\Bad{}$\BigO{8(\z\n)}$} &Partly\\
        \hline
        \end{tabular}
    }
\end{table}

\section[GeoBFT: Geo-Scale Consensus]{G\MakeLowercase{eo}BFT: Geo-Scale Consensus}\label{sec:geobft}

We now present our Geo-Scale Byzantine Fault-Tolerant consensus protocol (\GeoBFT{}) 
that uses topological information to group all replicas in a single region into a single cluster. 
Likewise, \GeoBFT{} assigns each client to a single cluster. 
This clustering helps in attaining high throughput and scalability in geo-scale deployments. 
\GeoBFT{} operates in rounds, and in each round, every cluster will be able to propose a single client request for execution. 
Next, we sketch the high-level working of such a round of \GeoBFT{}. 
Each round consists of the three steps sketched in Figure~\ref{fig:geobft-overview}: \emph{local replication}, \emph{global sharing}, and \emph{ordering and execution}, which 
we further detail next.

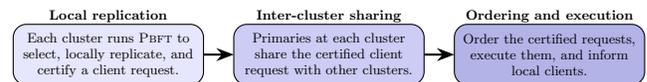
\begin{figure}[t]
    \centering
        \begin{tikzpicture}[scale=0.575,transform shape]
        \draw[draw=black,rounded corners=5pt,fill=black!10!blue!10] (1.2, 0.3) rectangle (5.6, 1.7);
        \draw[draw=black,rounded corners=5pt,fill=black!20!blue!20] (6.3, 0.3) rectangle (10.7, 1.7);
        \draw[draw=black,rounded corners=5pt,fill=black!30!blue!30] (11.4, 0.3) rectangle (15.8, 1.7);
        
        \node[align=center] at (3.4, 1) {Each cluster runs \PBFT{} to\\select, locally replicate, and\\certify a client request.};
        \node[align=center] at (8.5, 1) {Primaries at each cluster\\share the certified client\\request with other clusters.};
        \node[align=center] at (13.6, 1) {Order the certified requests,\\ execute them, and inform\\local clients.};

        \node[above=-2pt,font=\bfseries,align=center] at (3.4, 1.7) {Local replication};
        \node[above=-2pt,font=\bfseries,align=center] at (8.5, 1.7) {Inter-cluster sharing};
        \node[above=-2pt,font=\bfseries,align=center] at (13.6, 1.7) {Ordering and execution};
        
        \path[thick] (5.6,1) edge[->] (6.3,1) (10.7, 1) edge[->] (11.4,1);
    \end{tikzpicture}%
    \vspace{-4mm}
    \caption{Steps in a round of the \GeoBFT{} protocol.}
    \label{fig:geobft-overview}
\end{figure}

\begin{enumerate}
\item At the start of each round, each cluster chooses a single transaction of a local client. 
Next, each cluster \emph{locally replicates} its chosen transaction in a Byzantine fault-tolerant manner using \PBFT{}. 
At the end of successful local replication, \PBFT{} guarantees that each non-faulty replica can prove successful local replication via a \emph{commit certificate}.
\item Next, each cluster shares the locally-replicated transaction along with its commit certificate with all other clusters.  
To minimize inter-cluster communication, we use a novel \emph{optimistic global sharing protocol}. 
Our optimistic global sharing protocol has a global phase in which clusters exchange locally-replicated transactions, followed by a local phase in which clusters distribute any received transactions locally among all local replicas. To deal with failures, the global sharing protocol utilizes a novel remote view-change protocol.
\item Finally, after receiving all transactions that are locally-replicated in other clusters, 
each replica in each cluster can deterministically \emph{order} all these transactions and proceed with their \emph{execution}.  
After execution, the replicas in each cluster inform only local clients of the outcome of the execution of their transactions 
(e.g., confirm execution or return any execution results).
\end{enumerate}

In Figure~\ref{fig:geobft_sketch}, we sketch a single round of \GeoBFT{} in a setting of two clusters with four replicas each. 

\begin{figure}[t!]
    \centering
    \begin{tikzpicture}[yscale=0.4,xscale=1.05]
        \draw[draw=purple!20,fill=purple!20] (4, 0) rectangle (5, 10.5);
        
        \draw[thick,draw=black!75] (0.75, 4.5) edge[green!50!black!90] ++(6.75, 0)
                                   (0.75, 0) edge ++(6.75, 0)
                                   (0.75, 1) edge ++(6.75, 0)
                                   (0.75, 2) edge ++(6.75, 0)
                                   (0.75, 3) edge[blue!50!black!90] ++(6.75, 0) 

                                   (0.75, 10.5) edge[green!50!black!90] ++(6.75, 0)
                                   (0.75,  6) edge ++(6.75, 0)
                                   (0.75,  7) edge ++(6.75, 0)
                                   (0.75,  8) edge ++(6.75, 0)
                                   (0.75,  9) edge[blue!50!black!90] ++(6.75, 0);

       \draw[thin,draw=black!75] (1,   0) edge ++(0, 4.5)
                                  (2,   0) edge ++(0, 4.5)
                                  (4,   0) edge ++(0, 4.5)
                                  (5,   0) edge ++(0, 4.5)
                                  (6,   0) edge ++(0, 4.5)
                                  (6.5, 0) edge ++(0, 4.5)
                                  (7.5, 0) edge ++(0, 4.5)
                                  
                                  (1,   6) edge ++(0, 4.5)
                                  (2,   6) edge ++(0, 4.5)
                                  (4,   6) edge ++(0, 4.5)
                                  (5,   6) edge ++(0, 4.5)
                                  (6,   6) edge ++(0, 4.5)
                                  (6.5, 6) edge ++(0, 4.5)
                                  (7.5, 6) edge ++(0, 4.5);

        \node[left] at (0.8, 0) {$\Replica_{2,3}$};
        \node[left] at (0.8, 1) {$\Replica_{2,2}$};
        \node[left] at (0.8, 2) {$\Replica_{2,1}$};
        \node[left] at (0.8, 3) {$\PrimaryC{\Cluster_2}$};
        \node[left] at (0.8, 4.5) {$\Client_2$};
        
        \node[left] at (0.8, 6) {$\Replica_{1,3}$};
        \node[left] at (0.8, 7) {$\Replica_{1,2}$};
        \node[left] at (0.8, 8) {$\Replica_{1,1}$};
        \node[left] at (0.8, 9) {$\PrimaryC{\Cluster_1}$};
        \node[left] at (0.8, 10.5) {$\Client_1$};

        \path[->] (1, 4.5) edge node[above=-6pt,xshift=4pt,label] {$\Transaction_2$} (2, 3)
                  (1, 10.5) edge node[above=-6pt,xshift=4pt,label] {$\Transaction_1$} (2, 9)

                  (2, 3) edge (2.5, 2) edge (2.5, 1) edge (2.5, 0)
                  (2, 9) edge (2.5, 8) edge (2.5, 7) edge (2.5, 6);

        \path[<-] (4, 3) edge (3.5, 2) edge (3.5, 1) edge (3.5, 0)
                  (4, 2) edge (3.5, 3) edge (3.5, 1) edge (3.5, 0)
                  (4, 1) edge (3.5, 3) edge (3.5, 2) edge (3.5, 0)
                  (4, 0) edge (3.5, 3) edge (3.5, 2) edge (3.5, 1)
                  (4, 9) edge (3.5, 8) edge (3.5, 7) edge (3.5, 6)
                  (4, 8) edge (3.5, 9) edge (3.5, 7) edge (3.5, 6)
                  (4, 7) edge (3.5, 9) edge (3.5, 8) edge (3.5, 6)
                  (4, 6) edge (3.5, 9) edge (3.5, 8) edge (3.5, 7);
                  
        \path[->] (4, 9) edge (5, 3) edge (5, 2)
                  (4, 3) edge (5, 9) edge (5, 8)
                  
                  (5, 9) edge (6, 8) edge (6, 7) edge (6, 6)
                  (5, 8) edge (6, 9) edge (6, 7) edge (6, 6)

                  (5, 3) edge (6, 2) edge (6, 1) edge (6, 0)
                  (5, 2) edge (6, 3) edge (6, 1) edge (6, 0);

        \path[<-] (7.5, 10.5) edge (6.5, 9) edge (6.5, 8) edge (6.5, 7) edge (6.5, 6)
                  (7.5, 4.5) edge (6.5, 3) edge (6.5, 2) edge (6.5, 1) edge (6.5, 0);

        \draw[draw=orange!60,fill=orange!40,rounded corners] (2.2, 3.2) rectangle (3.8, -0.2)
                                                             (2.2, 9.2) rectangle (3.8, 5.8);

        \node[align=center,label] at (3, 1.5) {Local \PBFT{}\\Consensus\\on $\Transaction_2$};
        \node[align=center,label] at (3, 7.5) {Local \PBFT{}\\Consensus\\on $\Transaction_1$};
        
        \draw[draw=black!20,fill=black!10,rounded corners]  (6.05, -0.6) rectangle (6.45, 3.6) (6.05, 5.4) rectangle (6.45, 9.6);
        \node[align=center,label,rotate=270] at (6.25, 1.5) {Execute $\Transaction_1\Transaction_2$};
        \node[align=center,label,rotate=270] at (6.25, 7.5) {Execute $\Transaction_1\Transaction_2$};

        \node[align=center,below,label] at (1.5, 0) {Local\\Request};
        \node[align=center,below,label] at (  3, 0) {Local\\Replication};
        \node[align=center,below,label] at (4.5, 0) {Global\\Sharing};
        \node[align=center,below,label] at (5.5, 0) {Local\\Sharing};
        \node[align=center,below,label] at (  7, 0) {Local\\Inform};
        
        \draw[decoration={brace},decorate] (0.15, -0.4) -- (0.15, 3.4);
        \node[left=2pt] at (0.15, 1.5) {$\Cluster_2$};
        
        \draw[decoration={brace},decorate] (0.15, 5.6) -- (0.15, 9.4);
        \node[left=2pt] at (0.15, 7.5) {$\Cluster_1$};
    \end{tikzpicture}%
    \vspace{-4mm}
    \caption{Representation of the normal-case algorithm of \GeoBFT{} running on two clusters. Clients $\Client_i$, $i \in \{1,2\}$, request transactions $\Transaction_i$ from their local cluster $\Cluster_i$. The primary $\PrimaryC{\Cluster_i} \in \Cluster_i$ replicates this transaction to all local replicas using \PBFT{}. At the end of local replication, the primary can produce a cluster certificate for $\Transaction_i$. These are shared with other clusters via inter-cluster communication, after which all replicas in all clusters can execute $\Transaction_i$ and $\Cluster_i$ can inform $\Client_i$.}\label{fig:geobft_sketch}
\end{figure}
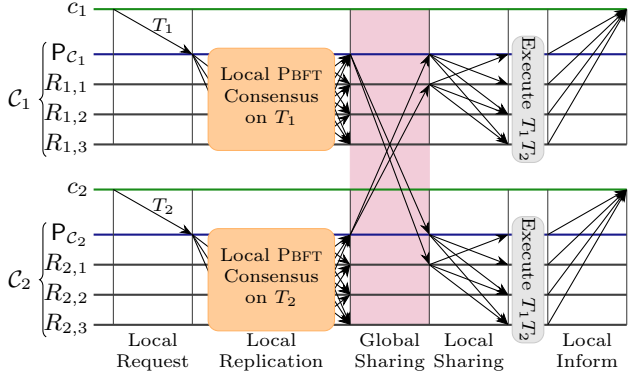

\subsection{Preliminaries}\label{ss:prelim}
To present \GeoBFT{} in detail, we first introduce the system model we use and the relevant notations.

Let $\Replicas$ be a set of replicas. We model a topological-aware \emph{system} as a partitioning of $\Replicas$ 
into a set of clusters $\Clusters = \{ \Cluster_1, \dots, \Cluster_{\z} \}$, 
in which each cluster $\Cluster_i$, $1 \leq i \leq \z$, is a set of $\abs{\Cluster_i} = \n$ replicas of which 
at most $\f$ are \emph{faulty} and can behave in \emph{Byzantine}, possibly coordinated and malicious, manners. 
We assume that in each cluster $\n > 3\f$.

\begin{remark}\label{rem:model}
We assumed $\z$ clusters with $\n > 3\f$ replicas each. Hence, $\n = 3\f + j$ for some $j \geq 1$. We use the same failure model as \STW{}~\cite{steward}, but our failure model differs from the more-general failure model utilized by \PBFT{}, \ZZ{}, and \HS{}~\cite{bft700,pbft,pbftj,zyzzyva,zyzzyvaj,hotstuff}. 
These protocols can each tolerate the failure of up-to-$\dsfloorfrac{\z\n}{3} = \dsfloorfrac{(3\f\z + \z j)}{3} = \f\z + \dsfloorfrac{\z j}{3}$ replicas, 
even if more than $\f$ of these failures happen in a single region; 
whereas \GeoBFT{} and \STW{} can only tolerate $\f\z$ failures, of which at most $\f$ can happen in a single cluster. E.g., if $\n = 13$, $\f = 4$, and $\z = 7$, then \GeoBFT{} and \STW{} can tolerate $\f\z=28$ replica failures in total, whereas the other protocols can tolerate $30$ replica failures. 
The failure model we use enables the efficient geo-scale aware design of \GeoBFT{}, this without facing well-known communication bounds~\cite{cryptorounds,netbound,dolevstrong,byzgen,interbound}.
\end{remark}

We write $\Faulty{\Cluster_i}$ to denote the Byzantine replicas in cluster $\Cluster_i$ and 
$\NonFaulty{\Cluster_i} = \Cluster_i \difference \Faulty{\Cluster_i}$ to denote the non-faulty replicas in $\Cluster_i$. 
Each replica $\Replica \in \Cluster_i$ has a unique identifier $\ID{\Replica}$, $1 \leq \ID{\Replica} \leq \n$. We assume that non-faulty replicas behave in accordance to the protocol and are deterministic: on identical inputs, all non-faulty replicas must produce identical outputs. We do not make any assumptions on clients: all client can be malicious without affecting \GeoBFT{}.

Some messages in \GeoBFT{} are forwarded (for example, the client request and commit certificates during inter-cluster sharing). 
To ensure that malicious replicas do not tamper with messages while forwarding them, 
we sign these messages using digital signatures~\cite{cryptobook,hac}. 
We write $\SignMessage{u}{m}$ to denote a message signed by $u$. 
We assume that it is practically impossible to forge digital signatures. 
We also assume \emph{authenticated communication}: Byzantine replicas can impersonate each other, 
but no replica can impersonate another non-faulty replica. 
Hence, on receipt of a message $m$ from replica $\Replica \in \Cluster_i$, one can determine that $\Replica$ did send $m$ if $\Replica \notin \Faulty{\Cluster_i}$;  and one can only determine that $m$ was sent by a non-faulty replica if $\Replica \in \NonFaulty{\Cluster_i}$. In the permissioned setting, authenticated communication is a minimal requirement to deal with Byzantine behavior, as otherwise Byzantine replicas can impersonate all non-faulty replicas (which would lead to so-called Sybil attacks)~\cite{sybil}. For messages that are forwarded, authenticated communication is already provided via digital signatures. For all other messages, we use less-costly message authentication codes~\cite{cryptobook,hac}. Replicas will discard any messages that are not well-formed, have invalid message authentication codes (if applicable), or have invalid signatures (if applicable).

Next, we define the consensus provided by \GeoBFT{}.
\begin{definition}\label{def:consensus}
Let $\Clusters$ be a system over $\Replicas$. A single run of any \emph{consensus protocol} should satisfy the following two requirements: 
\begin{description}[nosep]
    \item[Termination] Each non-faulty replica in $\Replicas$ executes a transaction.
    \item[Non-divergence] All non-faulty replicas execute the same transaction.
\end{description}

Termination is typically referred to as \emph{liveness}, whereas non-divergence is typically referred to as \emph{safety}. 
A single round of our \GeoBFT{} consists of $\z$ consecutive runs of the \PBFT{} consensus protocol. Hence, in a single round of \GeoBFT{}, all non-faulty replicas execute the same sequence of $\z$ transactions.
\end{definition}

To provide \emph{safety}, we do not need any other assumptions on communication or on the behavior of clients. Due to well-known impossibility results for asynchronous consensus~\cite{cap12,capthm,flp,capproof}, we can only provide \emph{liveness} in periods of \emph{reliable bounded-delay communication} during which all messages sent by non-faulty replicas will arrive at their destination within some maximum delay.

\subsection{Local Replication}\label{ss:local_rep}

In the first step of \GeoBFT{}, the local replication step, 
each cluster will independently choose a client request to execute. 
Let $\Clusters$ be a system. 
Each round $\rn$ of \GeoBFT{} starts with each cluster $\Cluster \in \Clusters$ replicating a client 
request $\Transaction$ of client $\Client \in \Clients{\Cluster}$. 
To do so, \GeoBFT{} relies on \PBFT{}~\cite{pbft,pbftj},\footnote{Other consensus protocols such as \ZZ{}~\cite{bft700,zyzzyva,zyzzyvaj} and \HS{}~\cite{hotstuff} promise to improve on \PBFT{} by sharply reducing communication. In our setting, where local communication is abundant (see Table~\ref{tbl:geoscale_cost}), such improvements are unnecessary, and the costs of \ZZ{} (reliable clients) and \HS{} (high computational complexity) can be avoided.} 
a primary-backup protocol in which one replica acts as the \emph{primary}, while all the other replicas act as \emph{backups}. 
In \PBFT{}, the primary is responsible for coordinating the replication of client transactions. 
We write $\PrimaryC{\Cluster}$ to denote the replica in $\Cluster$ that is the current \emph{local primary} of cluster $\Cluster$. The normal-case of \PBFT{} operates in four steps which we sketch in Figure~\ref{fig:pbft}. Next, we detail these steps.

First, the primary $\PrimaryC{\Cluster}$ receives client requests of the form $\SignMessage{\Client}{\Transaction}$, 
transactions $\Transaction$ signed by a local client $\Client \in \Clients{\Cluster}$. 

Then, in round $\rn$, $\PrimaryC{\Cluster}$ chooses a request $\SignMessage{\Client}{\Transaction}$ and 
initiates the replication of this request by proposing it to all replicas via a \Name{preprepare} message.  
When a backup replica receives a \Name{preprepare} message from the primary, it agrees to participate in 
a two-phase Byzantine commit protocol. 
This commit protocol can succeed if at least $\n-2\f$ non-faulty replicas receive the same \Name{preprepare} message. 

In the first phase of the Byzantine commit protocol, each replica $\Replica$ responds to the \Name{preprepare} message $m$
by broadcasting a \Name{prepare} message in support of $m$. 
After broadcasting the \Name{prepare} message, $\Replica$ waits until it receives $\n-\f$ \Name{prepare} 
messages in support of $m$ (indicating that at least $\n-2\f$ non-faulty replicas support $m$). 

Finally, after receiving these messages, $\Replica$ enters the second phase of the Byzantine commit protocol and broadcasts a \Name{commit} message in support of $m$. 
Once a replica $\Replica$ receives $\n-\f$ \Name{commit} messages in support of $m$, 
it has the guarantee that eventually all replicas will commit to $\SignMessage{\Client}{\Transaction}$. 

This protocol exchanges sufficient information among all replicas to enable detection 
of malicious behavior of the primary and to recover from any such behavior. Moreover, on success, each non-faulty replica $\Replica \in \Cluster$ will be committed to the 
proposed request $\SignMessage{\Client}{\Transaction}$  and will be able to construct a \emph{commit certificate} $\Cert{\Replica}{\SignMessage{\Client}{\Transaction}, \rn}$ that proves this commitment.  In \GeoBFT{}, this commit certificate consists of the client request $\SignMessage{\Client}{\Transaction}$ and 
$\n - \f > 2\f$ identical \Name{commit} messages for $\SignMessage{\Client}{\Transaction}$ signed by distinct replicas. 
Optionally, \GeoBFT{} can use threshold signatures to represent these $\n - \f$ signatures 
via a single constant-sized threshold signature~\cite{rsasign}.

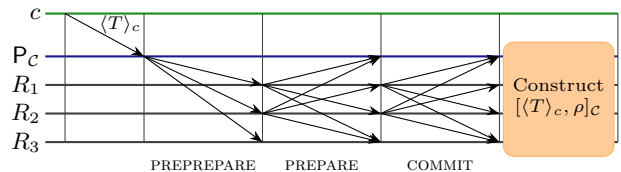
\begin{figure}[t!]
    \centering
    \begin{tikzpicture}[yscale=0.38,xscale=1.05]
        \draw[thick,draw=black!75] (0.75, 4.5) edge[green!50!black!90] ++(7.25, 0)
                                   (0.75,   0) edge ++(7.25, 0)
                                   (0.75,   1) edge ++(7.25, 0)
                                   (0.75,   2) edge ++(7.25, 0)
                                   (0.75,   3) edge[blue!50!black!90] ++(7.25, 0);

        \draw[thin,draw=black!75] (1,   0) edge ++(0, 4.5)
                                  (2, 0) edge ++(0, 4.5)
                                  (3.5,   0) edge ++(0, 4.5)
                                  (5, 0) edge ++(0, 4.5)
                                  (6.5,   0) edge ++(0, 4.5)
                                  (8,   0) edge ++(0, 4.5);

        \node[left] at (0.8, 0) {$\Replica_3$};
        \node[left] at (0.8, 1) {$\Replica_2$};
        \node[left] at (0.8, 2) {$\Replica_1$};
        \node[left] at (0.8, 3) {$\PrimaryC{\Cluster}$};
        \node[left] at (0.8, 4.5) {$\Client$};

        \path[->] (1, 4.5) edge node[above=-4pt,xshift=6pt,label] {$\SignMessage{\Client}{\Transaction}$} (2, 3)
                  (2, 3) edge (3.5, 2)
                         edge (3.5, 1)
                         edge (3.5, 0)
                           
                  (3.5, 2) edge (5, 0)
                           edge (5, 1)
                           edge (5, 3)

                  (3.5, 1) edge (5, 0)
                           edge (5, 2)
                           edge (5, 3)
                         
                  (5, 2) edge (6.5, 0)
                         edge (6.5, 1)
                         edge (6.5, 3)

                  (5, 1) edge (6.5, 0)
                         edge (6.5, 2)
                         edge (6.5, 3)
                           ;
                           
        \draw[draw=orange!60,fill=orange!40,rounded corners]  (6.55, -0.5) rectangle (7.95, 3.5);
        \node[align=center,label] at (7.25, 1.5) {Construct\\$\Cert{\Cluster}{\SignMessage{\Client}{\Transaction}, \rn}$};

        \node[below,label] at (2.75, 0) {\Name{preprepare}};
        \node[below,label] at (4.25, 0) {\Name{prepare}};
        \node[below,label] at (5.75, 0) {\Name{commit}};
    \end{tikzpicture}%
    \vspace{-4mm}
    \caption{The normal-case working of round $\rn$ of \PBFT{} within a cluster $\Cluster$: a client $\Client$ requests transaction $\Transaction$, the primary $\PrimaryC{\Cluster}$ proposes this request to all local replicas, which prepare and commit this proposal, and, finally, all replicas can construct a commit certificate.}\label{fig:pbft}
\end{figure}

In \GeoBFT{}, we use a \PBFT{} implementation that only uses digital signatures for client requests and \Name{commit} messages, as these are the only messages that need forwarding. In this configuration, \PBFT{} provides the following properties:

\begin{lemma}[Castro et al.~\cite{pbft,pbftj}]\label{lem:pbft}
Let $\Clusters$ be a system and let $\Cluster \in \Clusters$ be a cluster with $\n > 3\f$. We have the following:
\begin{description}[nosep]
    \item[Termination] If communication is reliable, has bounded delay, and a replica $\Replica \in \Cluster$ is able to construct a commit certificate $\Cert{\Replica}{\SignMessage{\Client}{\Transaction}, \rn}$, then all non-faulty replicas $\Replica' \in \NonFaulty{\Cluster}$ will eventually be able to construct a commit certificate $\Cert{\Replica'}{\SignMessage{\Client'}{\Transaction'}, \rn}$.
    \item[Non-divergence] If replicas $\Replica_1, \Replica_2 \in \Cluster$ are able to construct commit certificates $\Cert{\Replica_1}{\SignMessage{\Client_1}{\Transaction_1}, \rn}$ and $\Cert{\Replica_2}{\SignMessage{\Client_2}{\Transaction_2}, \rn}$, respectively, then $\Transaction_1 = \Transaction_2$ and $\Client_1 = \Client_2$.
\end{description}
\end{lemma}

From Lemma~\ref{lem:pbft}, we conclude that all commit certificates made by replicas in $\Cluster$ for round $\rn$ 
show commitment to the same client request $\SignMessage{\Client}{\Transaction}$. 
Hence, we write $\Cert{\Cluster}{\SignMessage{\Client}{\Transaction}, \rn}$, to represent a 
commit certificate from some replica in cluster $\Cluster$.

To guarantee the correctness of \PBFT{} (Lemma~\ref{lem:pbft}), we need to prove that both non-divergence and termination hold.
From the normal-case working outlined above and in Figure~\ref{fig:pbft}, \PBFT{} guarantees non-divergence 
independent of the behavior of the primary or any malicious replicas. 

To guarantee termination when communication is reliable and has bounded delay, \PBFT{} uses \emph{view-changes} and \emph{checkpoints}. If the primary is faulty and prevents any replica from making progress, then the \emph{view-change protocol} enables non-faulty replicas to reliably detect primary failure, recover a common non-divergent state, and trigger primary replacement until a non-faulty primary is found. After a successful view-change, progress is resumed. We refer to these \PBFT{}-provided view-changes as \emph{local view-changes}. The \emph{checkpoint protocol} enables non-faulty replicas to recover from failures and malicious behavior that do not trigger a view-change.

\subsection{Inter-Cluster Sharing}\label{ss:intershare}
Once a cluster has completed local replication of a client request, it proceeds with the second step: sharing the client request with all other clusters. Let $\Clusters$ be a system and $\Cluster \in \Clusters$ be a cluster. After $\Cluster$ reaches local consensus on client request $\SignMessage{\Client}{\Transaction}$ in round $\rn$---enabling construction of the commit certificate $\Cert{\Cluster}{\SignMessage{\Client}{\Transaction}, \rn}$ that proves local consensus---$\Cluster$ needs to exchange this client request and the accompanying proof with all other clusters.  
This exchange step requires global inter-cluster communication, which we want to minimize while retaining 
the ability to reliably detect failure of the sender. 
However, minimizing this inter-cluster communication is not as straightforward as it sounds, which we illustrate next:

\begin{example}\label{ex:sopti_inters}
Let $\Clusters$ be a system with two clusters $\Cluster_1, \Cluster_2 \in \Clusters$. 
Consider a simple global communication protocol in which a message $m$ is sent from $\Cluster_1$ to $\Cluster_2$ by 
requiring the primary $\PrimaryC{\Cluster_1}$ to send $m$ to the primary $\PrimaryC{\Cluster_2}$ 
(which can then disseminate $m$ in $\Cluster_2$). 
In this protocol, the replicas in $\Cluster_2$ cannot determine what went wrong if they do not receive any messages. To show this, we distinguish two cases:
\begin{enumerate}[wide,nosep,label=(\arabic*)]
\item $\PrimaryC{\Cluster_1}$ is Byzantine and behaves correctly toward every replica, except that it never sends messages to $\PrimaryC{\Cluster_2}$, while $\PrimaryC{\Cluster_2}$ is non-faulty.
\item $\PrimaryC{\Cluster_1}$ is non-faulty, while $\PrimaryC{\Cluster_2}$ is Byzantine and behaves correctly toward every replica, except that it drops all messages sent by $\PrimaryC{\Cluster_1}$.
\end{enumerate}

In both cases, the replicas in $\Cluster_2$ do not receive any messages from $\Cluster_1$, while both clusters see correct behavior of their primaries with respect to local consensus. Indeed, with this little amount of communication, it is impossible for replicas in $\Cluster_2$ to determine whether $\PrimaryC{\Cluster_1}$ is faulty (and did not send any messages) or $\PrimaryC{\Cluster_2}$ is faulty (and did not forward any received messages from $\Cluster_1$). 
\end{example}

In \GeoBFT{}, we employ an \emph{optimistic} approach to reduce communication among the clusters. 
Our optimistic approach consists of a low-cost normal-case protocol that will succeed when communication is reliable and 
the primary of the sending cluster is non-faulty. 
To deal with any failures, we use a \emph{remote view-change} protocol that guarantees eventual normal-case behavior when communication is reliable. 
First, we describe the normal-case protocol, after which we will describe in detail the remote view-change protocol.

\paragraph*{Optimistic inter-cluster sending}
In the optimistic case, where participants are non-faulty, we want to send a minimum number of messages while retaining the ability to reliably detect failure of the sender. In Example~\ref{ex:sopti_inters}, we already showed that sending only a single message is not sufficient. Sending $\f+1$ messages is sufficient, however.

Let $m = (\SignMessage{\Client}{\Transaction}, \Cert{\Cluster_1}{\SignMessage{\Client}{\Transaction}, \rn})$ 
be the message that some replica in cluster $\Cluster_1$ needs to send to some replicas $\Cluster_2$. 
Note that $m$ includes the request replicated in $\Cluster_1$ in round $\rn$, and the commit-certificate, 
which is the proof that such a replication did take place.
Based on the observations made above, we propose a two-phase normal-case global sharing protocol. 
We sketch this normal-case sending protocol in Figure~\ref{fig:nc_send_sketch} and 
present the detailed pseudo-code for this protocol in Figure~\ref{fig:nc_send}.

\begin{figure}[t!]
    \centering
    \begin{tikzpicture}[yscale=0.38,xscale=1.05]
            \draw[draw=purple!20,fill=purple!20] (1, 0) rectangle (3, 4.5);
    
        \draw[thick,draw=black!75] (0.75, 4.5) edge[green!50!black!90] ++(4.25, 0)
                                   (0.75,   0) edge ++(4.25, 0)
                                   (0.75,   1) edge ++(4.25, 0)
                                   (0.75,   2) edge ++(4.25, 0)
                                   (0.75,   3) edge[blue!50!black!90] ++(4.25, 0);

        \draw[thin,draw=black!75] (1,   0) edge ++(0, 4.5)
                                  (3, 0) edge ++(0, 4.5)
                                  (5,  0) edge ++(0, 4.5);

        \node[left] at (0.8, 0) {$\Replica_{2,3}$};
        \node[left] at (0.8, 1) {$\Replica_{2,2}$};
        \node[left] at (0.8, 2) {$\Replica_{2,1}$};
        \node[left] at (0.8, 3) {$\PrimaryC{\Cluster_2}$};
        \node[left] at (0.8, 4.5) {$\PrimaryC{\Cluster_1}$};

        \path[->] (1, 4.5) edge (3, 3)
                           edge (3, 2)
                  
                  (3, 3) edge (5, 2) edge (5, 1) edge (5, 0)
                  (3, 2) edge (5, 3) edge (5, 1) edge (5, 0)
                           ;

        \node[below,label] at (2, 0) {Global phase};
        \node[below,label] at (4, 0) {Local phase};

        \draw[decoration={brace},decorate] (0.15, -0.4) -- (0.15, 3.4);
        \node[left=2pt] at (0.15, 1.5) {$\Cluster_2$};
    \end{tikzpicture}%
    \vspace{-4mm}
    \caption{A schematic representation of the normal-case working of the global sharing protocol used by $\Cluster_1$ to send $m = (\SignMessage{\Client}{\Transaction}, \Cert{\Cluster_1}{\SignMessage{\Client}{\Transaction}, \rn})$ to $\Cluster_2$.}\label{fig:nc_send_sketch}
\end{figure}

\begin{figure}[t!]
    \begin{myprotocol}
        \TITLE{The global phase}{used by the primary $\PrimaryC{\Cluster_1}$}
        \STATE Choose a set $S$ of $\f+1$ replicas in $\Cluster_2$.
        \STATE Send $m$ to each replica in $S$.\label{fig:nc_send:send}
        \SPACE
        \TITLE{The local phase}{used by replicas $\Replica \in \Cluster_2$}
        \EVENT{receive $m$ from a replica $\Replica[q] \in \Cluster_1$}
            \STATE Broadcast $m$ to all replicas in $\Cluster_2$.\label{fig:nc_send:forward}
        \ENDEVENT
    \end{myprotocol}
    \caption{The normal-case global sharing protocol used by $\Cluster_1$ to send $m = (\SignMessage{\Client}{\Transaction}, \Cert{\Cluster_1}{\SignMessage{\Client}{\Transaction}, \rn})$ to $\Cluster_2$.}\label{fig:nc_send}
\end{figure}

In the \emph{global phase}, the primary $\PrimaryC{\Cluster_1}$ sends $m$ to $\f+1$ replicas in $\Cluster_2$. 
In the \emph{local phase}, each non-faulty replica $\Replica \in \NonFaulty{\Cluster_2}$ that 
receives a well-formed $m$ forwards $m$ to all replicas in its cluster $\Cluster_2$. 

\begin{proposition}\label{prop:nc_send}
Let $\Clusters$ be a system, let $\Cluster_1, \Cluster_2 \in \Clusters$ be two clusters, and let $m = (\SignMessage{\Client}{\Transaction}, \Cert{\Cluster_1}{\SignMessage{\Client}{\Transaction}, \rn})$ be the message $\Cluster_1$ sends to $\Cluster_2$ using the normal-case global sharing protocol of Figure~\ref{fig:nc_send}. We have the following:
\begin{description}[nosep]
    \item[Receipt] If the primary $\PrimaryC{\Cluster_1}$ is non-faulty and communication is reliable, then every replica in $\Cluster_2$ will eventually receive $m$.
    \item[Agreement] Replicas in $\Cluster_2$ will only accept client request $\SignMessage{\Client}{\Transaction}$ from $\Cluster_1$ in round $\rn$.
\end{description}
\end{proposition}
\begin{proof}
If the primary $\PrimaryC{\Cluster_1}$ is non-faulty and communication is reliable, 
then $\f+1$ replicas in $\Cluster_2$ will receive $m$ (Line~\ref{fig:nc_send:send}). As at most $\f$ replicas in $\Cluster_2$ are Byzantine, at least one of these receiving replicas is non-faulty and will forward this message $m$ to all replicas in $\Cluster_2$ (Line~\ref{fig:nc_send:forward}), proving termination.

The commit certificate $\Cert{\Cluster_1}{\SignMessage{\Client}{\Transaction}, \rn}$ cannot be forged by faulty replicas, 
as it contains signed $\Name{commit}$ messages from $\n-\f > \f$ replicas. Hence, the integrity of any message $m$ forwarded by replicas in $\Cluster_2$ can easily be verified. Furthermore, Lemma~\ref{lem:pbft} rules out the existence of any other messages $m' =\Cert{\Cluster_1}{\SignMessage{\Client'}{\Transaction'}, \rn}$, proving agreement.
\end{proof}

We notice that there are two cases in which replicas in $\Cluster_2$ 
do not receive $m$ from $\Cluster_1$: either $\PrimaryC{\Cluster_1}$ is faulty and 
did not send $m$ to $\f+1$ replicas in $\Cluster_2$, or communication is unreliable, and messages are delayed or lost. 
In both cases, non-faulty replicas in $\Cluster_2$ initiate \emph{remote view-change} to force primary replacement in $\Cluster_1$ (causing replacement of the primary $\PrimaryC{\Cluster_1}$).

\paragraph*{Remote view-change}
The normal-case global sharing protocol outlined will only succeed if communication is reliable and the primary of the sending cluster is non-faulty. 
To recover from any failures, we provide a remote view-change protocol. Let $\Clusters = \{ \Cluster_1, \dots, \Cluster_{\z} \}$ be a system. 
To simplify presentation, we focus on the case in which the primary of cluster $\Cluster_1$ fails to send $m = (\SignMessage{\Client}{\Transaction}, \Cert{\Cluster_1}{\SignMessage{\Client}{\Transaction}, \rn})$ to replicas of $\Cluster_2$. Our remote view-change protocol consists of four phases, which we detail next.

First, non-faulty replicas in cluster $\Cluster_2$ detect the failure of the current primary $\PrimaryC{\Cluster_1}$ 
of $\Cluster_1$ to send $m$.
Note that although the replicas in $\Cluster_2$ have no information about the contents of message $m$, they are 
awaiting arrival of a well-formed message $m$ from $\Cluster_1$ in round $\rn$. 
Second, the non-faulty replicas in $\Cluster_2$ initiate agreement on failure detection. 
Third, after reaching agreement, the replicas in $\Cluster_2$ send their request for a 
remote view-change to the replicas in $\Cluster_1$ in a reliable manner. 
In the fourth and last phase, the non-faulty replicas in $\Cluster_1$ trigger a local view-change, replace $\PrimaryC{\Cluster_1}$, 
and instruct the new primary to resume global sharing with $\Cluster_2$. 
Next, we explain each phase in detail.

To be able to detect failure, $\Cluster_2$ must assume reliable communication with bounded delay. This allows the usage of timers to detect failure. To do so, every replica $\Replica \in \Cluster_2$ sets a timer for $\Cluster_1$  at the start of round $\rn$ and waits until it receives a valid message $m$ from $\Cluster_1$. If the timer expires before $\Replica$ receives such an $m$, then $\Replica$ detects failure of $\Cluster_1$ in round $\rn$.  Successful detection will eventually lead to a remote view-change request.

From the perspective of $\Cluster_1$, remote view-changes are controlled by external parties. This leads to several challenges not faced by traditional \PBFT{} view-changes (the local view-changes used within clusters, e.g., as part of local replication):
\begin{enumerate}[wide,nosep,label=(\arabic*)]
\item A remote view-change  in $\Cluster_1$ requested by $\Cluster_2$ should only trigger at most a single local view-change in $\Cluster_1$, otherwise remote view-changes enable \emph{replay attacks}.
\item While replicas in $\Cluster_1$ detect failure of $\PrimaryC{\Cluster_1}$ and initiate local view-change, it is possible that $\Cluster_2$ detects failure of $\Cluster_1$ and requests remote view-change in $\Cluster_1$. In this case, only a single successful view-change in $\Cluster_1$ is necessary.
\item Likewise, several clusters $\Cluster_2, \dots, \Cluster_{\z}$ can simultaneously detect failure of $\Cluster_1$ and request remote view-change in $\Cluster_1$. Also in this case, only a single successful view-change in $\Cluster_1$ is necessary.
\end{enumerate}

Furthermore, a remote view-change request for cluster $\Cluster_1$ cannot depend on any information 
only available to $\Cluster_1$ (e.g., the current primary $\PrimaryC{\Cluster_1}$ of $\Cluster_1$). 
Likewise, the replicas in $\Cluster_1$ cannot determine which messages (for which rounds) 
have already been sent by previous (possibly malicious) primaries of $\Cluster_1$: remote view-change requests must include this information. 
Our remote view-change protocol addresses each of these concerns. In Figures~\ref{fig:rvc_sketch} and~\ref{fig:rvc}, we sketch this protocol and its pseudo-code. Next, we describe the protocol in detail.

\begin{figure}[t!]
    \centering
    \begin{tikzpicture}[yscale=0.4,xscale=1.05]
        \draw[draw=purple!20,fill=purple!20] (2.5, 0) rectangle (4, 7.5);
        \draw[thick,draw=black!75] (0.75, 4.5) edge ++(6.75, 0)
                                   (0.75, 5.5) edge ++(6.75, 0)
                                   (0.75, 6.5) edge ++(6.75, 0)
                                   (0.75, 7.5) edge ++(6.75, 0)
                                   (0.75,   0) edge ++(6.75, 0)
                                   (0.75,   1) edge ++(6.75, 0)
                                   (0.75,   2) edge ++(6.75, 0)
                                   (0.75,   3) edge ++(6.75, 0);

        \draw[thin,draw=black!75] (1, 0) edge ++(0, 3)
                                  (2.5, 0) edge ++(0, 7.5)
                                  
                                  (4, 0) edge ++(0, 7.5)
                                  (5.5, 4.5) edge ++(0, 3)
                                  (7.5, 4.5) edge ++(0, 3)
                                  ;

        \node[left] at (0.8, 0) {$\Replica_{2,4}$};
        \node[left] at (0.8, 1) {$\Replica_{2,3}$};
        \node[left] at (0.8, 2) {$\Replica_{2,2}$};
        \node[left] at (0.8, 3) {$\Replica_{2,1}$};

        \node[left] at (0.8, 4.5) {$\Replica_{1,4}$};
        \node[left] at (0.8, 5.5) {$\Replica_{1,3}$};
        \node[left] at (0.8, 6.5) {$\Replica_{1,2}$};
        \node[left] at (0.8, 7.5) {$\Replica_{1,1}$};

        \path[->] (1, 0) edge (2.5, 1) edge (2.5, 2) edge (2.5, 3)
                  (1, 1) edge (2.5, 0) edge (2.5, 2) edge (2.5, 3)
                  (1, 2) edge (2.5, 0) edge (2.5, 1) edge (2.5, 3)
                  (1, 3) edge (2.5, 0) edge (2.5, 1) edge (2.5, 2)

                  (2.5, 0) edge (4, 4.5)
                  (2.5, 1) edge (4, 5.5)
                  (2.5, 2) edge (4, 6.5)
                  (2.5, 3) edge (4, 7.5)
                  
                  (4, 4.5) edge (5.5, 5.5) edge (5.5, 6.5) edge (5.5, 7.5)
                  (4, 5.5) edge (5.5, 4.5) edge (5.5, 6.5) edge (5.5, 7.5)
                  (4, 6.5) edge (5.5, 4.5) edge (5.5, 5.5) edge (5.5, 7.5)
                  (4, 7.5) edge (5.5, 4.5) edge (5.5, 5.5) edge (5.5, 6.5)
                  
         ;
         
        \draw[draw=orange!60,fill=orange!40,rounded corners]  (5.55, 4) rectangle (7.45, 8);
        \node[align=center,label] at (6.5, 6) {Detection \&\\view-change\\in $\Cluster_1$ (\PBFT{})};

        \node[below,label,align=center] at (1.75, 0) {\Name{DRvc}};
        \node[below,label,align=center] at (3.25, 0) {\Name{Rvc}};
        \node[below,label,align=center] at (4.75, 0) {(forward)};
        
        \node[above,label,align=center] at (1, 8.25) {Detection\\(in $\Cluster_2$)} edge[thick,->] (1, 3);
        \node[above,label,align=center] at (2.5, 8.25) {Agreement\\(in $\Cluster_2$)}  edge[thick,->] (2.5, 3);
        \node[above,label,align=center] at (4, 8.25) {Request\\view-change};

        \draw[decoration={brace},decorate,thick] (2.5, 8) -- (5.5, 8);

        \draw[decoration={brace},decorate] (0.15, -0.4) -- (0.15, 3.4);
        \node[left=2pt] at (0.15, 1.5) {$\Cluster_2$};
        
        \draw[decoration={brace},decorate] (0.15, 4.1) -- (0.15, 7.9);
        \node[left=2pt] at (0.15, 6) {$\Cluster_1$};
    \end{tikzpicture}%
    \vspace{-4mm}
    \caption{A schematic representation of the remote view-change protocol of \GeoBFT{} running at a system $\Clusters$ over $\Replicas$. This protocol is triggered when a cluster $\Cluster_2 \in \Clusters$ expects a message from $\Cluster_1 \in \Clusters$, but does not receive this message in time.}\label{fig:rvc_sketch}
\end{figure}
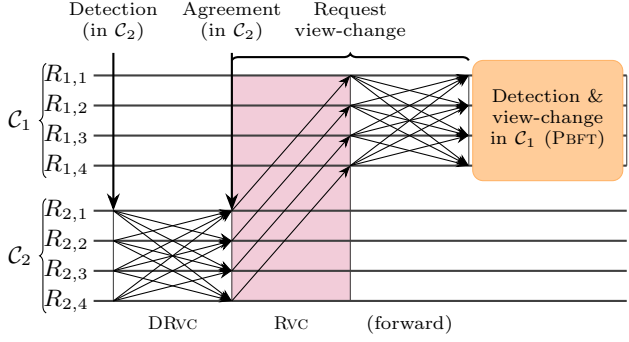

\begin{figure}[t!]
    \begin{myprotocol}
        \TITLE{Initiation role}{used by replicas $\Replica \in \Cluster_2$}
        \STATE $v_1 \GETS 0$ (number of remote view-changes in $\Cluster_1$ requested by $\Replica$).
        \EVENT{detect failure of $\Cluster_1$ in round $\rn$}
            \STATE Broadcast $\Message{DRvc}{\Cluster_1, \rn, v_1}$ to all replicas in $\Cluster_2$.\label{fig:rvc:drvc}
            \STATE $v_1 \GETS v_1 + 1$.
        \ENDEVENT
        \EVENT{$\Replica$ receives $\Message{DRvc}{\Cluster_1, \rn, v_1}$ from $\Replica' \in \Cluster_2$}\label{fig:rvc:resend}
            \IF{$\Replica$ received $(\SignMessage{\Client}{\Transaction}, \Cert{\Cluster}{\SignMessage{\Client}{\Transaction}, \rn})$ from $\Replica[q] \in \Cluster_1$}
                \STATE Send $(\SignMessage{\Client}{\Transaction}, \Cert{\Cluster}{\SignMessage{\Client}{\Transaction}, \rn})$ to $\Replica'$.
            \ENDIF
        \ENDEVENT
        \EVENT{$\Replica$ receives $\Message{DRvc}{\Cluster_1, \rn, v_1'}$ from $\f+1$ replicas in $\Cluster_2$}\label{fig:rvc:strength}
            \IF{$v_1 \leq v_1'$}
                \STATE $v_1 \GETS v_1'$.
                \STATE Detect failure of $\Cluster_1$ in round $\rn$ (if not yet done so).
            \ENDIF
        \ENDEVENT
        \EVENT{$\Replica$ receives $\Message{DRvc}{\Cluster_1, \rn, v_1}$ from $\n-\f$ replicas in $\Cluster_2$}\label{fig:rvc:request}
            \STATE Send $\SignMessage{\Replica}{\Message{Rvc}{\Cluster_1, \rn, v_1}}$ to $\Replica[q] \in \Cluster_1$, $\ID{\Replica} = \ID{\Replica[q]}$.\label{fig:rvc:rvcs}
        \ENDEVENT
        \SPACE
        \TITLE{Response role}{used by replicas $\Replica[q] \in \Cluster_1$}
        \EVENT{$\Replica[q]$ receives $\SignMessage{\Replica}{\Message{Rvc}{\Cluster_1, \rn, v}}$ from $\Replica$, $\Replica \in (\Replicas \difference \Cluster_1)$}\label{fig:rvc:bc1}
            \STATE Broadcast $\SignMessage{\Replica}{\Message{Rvc}{\Cluster_1, \rn, v}}$ to all replicas in $\Cluster_1$.        
        \ENDEVENT
        \EVENT{$\Replica[q]$ receives $\SignMessage{\Replica_i}{\Message{Rvc}{\Cluster_1, \rn, v}}$, $1 \leq i \leq \f+1$, such that:
            \begin{enumerate}[nosep]
                \item $\{ \Replica_i \mid 1\leq i \leq \f+1\} \subset \Cluster'$, $\Cluster' \in \Clusters$;
                \item $\abs{\{ \Replica_i \mid 1\leq i \leq \f+1\}} = \f+1$;
                \item no recent local view-change was triggered; and
                \item $\Cluster'$ did not yet request a $v$-th remote view-change
            \end{enumerate}}\label{fig:rvc:trigger}
            \STATE Detect failure of $\PrimaryC{\Cluster_1}$ (if not yet done so).
        \ENDEVENT
    \end{myprotocol}
    \caption{The remote view-change protocol of \GeoBFT{} running at a system $\Clusters$ over $\Replicas$. This protocol is triggered when a cluster $\Cluster_2 \in \Clusters$ expects a message from $\Cluster_1 \in \Clusters$, but does not receive this message in time.}\label{fig:rvc}
\end{figure}

Let $\Replica \in \Cluster_2$ be a replica that detects failure of $\Cluster_1$ in round $\rn$ and 
has already requested $v_1$ remote view-changes in $\Cluster_1$. 
Once a replica $\Replica$ detects a failure, it initiates the process of reaching 
an agreement on this failure among other replicas of its cluster $\Cluster_2$.
It does so by broadcasting message $\Message{DRvc}{\Cluster_1, \rn, v_1}$ to all replicas in 
$\Cluster_2$ (Line~\ref{fig:rvc:drvc} of Figure~\ref{fig:rvc}). 

Next, $\Replica$ waits until it receives identical $\Message{DRvc}{\Cluster_1, \rn, v_1}$ messages from $\n - \f$ distinct replicas in $\Cluster_2$ (Line~\ref{fig:rvc:request} of Figure~\ref{fig:rvc}). 
This guarantees that there is agreement among the non-faulty replicas in $\Cluster_2$ that $\Cluster_1$ has failed. 
After receiving these $\n-\f$ messages, $\Replica$ requests a remote view-change by 
sending message $\SignMessage{\Replica}{\Message{Rvc}{\Cluster_1, \rn, v_1}}$ to the 
replica $\Replica[q] \in \Cluster_1$ with $\ID{\Replica} = \ID{\Replica[q]}$ (Line~\ref{fig:rvc:rvcs} of Figure~\ref{fig:rvc}).

In case some other replica $\Replica' \in \Cluster_2$ received $m$ from $\Cluster_1$, then
$\Replica'$ would respond with message $m$ in response to the message $\Message{DRvc}{\Cluster_1, \rn, v}$  
(Line~\ref{fig:rvc:resend} of Figure~\ref{fig:rvc}). 
This allows $\Replica$ to recover in cases where it could not reach an agreement on the failure of $\Cluster_1$.
Finally, some replica $\Replica' \in \Cluster_2$ may detect the failure of $\Cluster_1$ later than $\Replica$. 
To handle such a case, we require each replica $\Replica'$ that receives identical $\Message{DRvc}{\Cluster_1, \rn, v}$ 
messages from $\f+1$ distinct replicas in $\Cluster_2$ to assume that the cluster $\Cluster_1$ has failed.
This assumption is valid as one of these $\f+1$ messages must have come from a 
non-faulty replica in $\Cluster_2$, which must have detected the failure of cluster $\Cluster_1$ successfully 
(Line~\ref{fig:rvc:strength} of Figure~\ref{fig:rvc}).

If replica $\Replica[q] \in \Cluster_1$ receives a remote view-change request 
$m_{\Name{Rcv}} = \SignMessage{\Replica}{\Message{Rvc}{\rn, v}}$ from $\Replica \in \Cluster_2$, 
then $\Replica[q]$ verifies whether $m_{\Name{Rcv}}$ is well-formed. 
If $m_{\Name{Rcv}}$ is well-formed, $\Replica[q]$ forwards $m_{\Name{Rcv}}$ to all replicas in 
$\Cluster_1$ (Line~\ref{fig:rvc:bc1} of Figure~\ref{fig:rvc}). 
Once $\Replica[q]$ receives $\f + 1$ messages identical to $m_{\Name{Rcv}}$, 
signed by distinct replicas in $\Cluster_2$, it concludes that at least one of 
these remote view-change requests must have come from a non-faulty replica in $\Cluster_2$. 
Next, $\Replica[q]$ determines whether it will honor this remote view-change request, which $\Replica[q]$ will do when no concurrent local view-change is in progress and when this is the first $v$-th remote view-change requested by $\Cluster_2$ (the lather prevents replay attacks). 
If these conditions are met, $\Replica[q]$ detects its current primary $\PrimaryC{\Cluster_1}$ as faulty (Line~\ref{fig:rvc:trigger} of Figure~\ref{fig:rvc}). 

When communication is reliable, the above protocol ensures that all non-faulty replicas in $\Cluster_1$ 
will detect failure of $\PrimaryC{\Cluster_1}$. 
Hence, eventually a successful local view-change will be triggered in $\Cluster_1$.
When a new primary in $\Cluster_1$ is elected, it takes one of the remote view-change requests it received and 
determines the rounds for which it needs to send requests 
(using the normal-case global sharing protocol of Figure~\ref{fig:nc_send}). 
As replicas in $\Cluster_2$ do not know the exact communication delays, they use exponential back off to determine the timeouts used while detecting subsequent failures of $\Cluster_1$.

We are now ready to prove the main properties of remote view-changes.

\begin{proposition}\label{prop:remote_vc}
Let $\Clusters$ be a system, let $\Cluster_1, \Cluster_2 \in \Clusters$ be two clusters, and let $m = (\SignMessage{\Client}{\Transaction}, \Cert{\Cluster}{\SignMessage{\Client}{\Transaction}, \rn})$ be the message $\Cluster_1$ needs to send to $\Cluster_2$ in round $\rn$. If communication is reliable and has bounded delay, then either every replica in $\Cluster_2$  will receive $m$ or $\Cluster_1$ will perform a local view-change.
\end{proposition}
\begin{proof}
Consider the remote view-change protocol of Figure~\ref{fig:rvc}. If a non-faulty replica $\Replica'\in \NonFaulty{\Cluster_2}$ receives $m$, then any replica in $\Cluster_2$ that did not receive $m$ will receive $m$ from $\Replica'$ (Line~\ref{fig:rvc:resend}). In all other cases, at least $\f+1$ non-faulty replicas in $\Cluster_2$ will not receive $m$ and will timeout. Due to exponential back-off, eventually each of these $\f+1$ non-faulty replicas will initiate and agree on the same $v_1$-th remote view-change. Consequently, all non-faulty replicas in $\NonFaulty{\Cluster_2}$ will participate in this remote view-change (Line~\ref{fig:rvc:strength}). As $\abs{\NonFaulty{\Cluster_2}} = \n - \f$, each of these $\n-\f$ replicas $\Replica \in \NonFaulty{\Cluster_2}$ will send $\SignMessage{\Replica}{\Message{Rvc}{\Cluster_1, \rn, v}}$ to some replica $\Replica[q] \in \Cluster_1$, $\ID{\Replica} = \ID{\Replica[q]}$ (Line~\ref{fig:rvc:request}). Let $S = \{ \Replica[q] \in \Cluster_1 \mid \Replica \in \NonFaulty{\Cluster_2} \land \ID{\Replica} = \ID{\Replica[q]} \}$ be the set of receivers in $\Cluster_1$ of these messages and let $T = S \intersect \NonFaulty{\Cluster_1}$. We have $\abs{S} = \n -\f > 2\f$ and, hence, $\abs{T} > \f$. Each replica $\Replica[q] \in T$ will broadcast the message it receives to all replicas in $\Cluster_1$ (Line~\ref{fig:rvc:bc1}). As $\abs{T} > \f$, this eventually triggers a local view-change in $\Cluster_1$ (Line~\ref{fig:rvc:trigger}).
\end{proof}

Finally, we use the results of Proposition~\ref{prop:nc_send} and Proposition~\ref{prop:remote_vc} to conclude

\begin{theorem}\label{thm:learning_z}
Let $\Clusters = \{\Cluster_1, \dots, \Cluster_{\z} \}$ be a system over $\Replicas$. If communication is reliable and has bounded delay, then every replica $\Replica \in \Replicas$ will, in round $\rn$, receive a set $\{ (\SignMessage{\Client_i}{\Transaction_i}, \Cert{\Cluster_i}{\SignMessage{\Client_i}{\Transaction_i}, \rn}) \mid (1 \leq i \leq \z) \land (\Client_i \in \Clients{\Cluster_i}) \}$ of $\z$ messages. These sets all contain identical client requests.
\end{theorem}
\begin{proof}
Consider cluster $\Cluster_i \in \Clusters$. If $\PrimaryC{\Cluster_i}$ behaves reliable, then Proposition~\ref{prop:nc_send} already proves the statement with respect to $(\SignMessage{\Client_i}{\Transaction_i}, \Cert{\Cluster_i}{\SignMessage{\Client_i}{\Transaction_i}, \rn}) $. Otherwise, if $\PrimaryC{\Cluster_i}$ behaves Byzantine, then then Proposition~\ref{prop:remote_vc} guarantees that either all replicas in $\Replicas$ will receive $(\SignMessage{\Client_i}{\Transaction_i}, \Cert{\Cluster_i}{\SignMessage{\Client_i}{\Transaction_i}, \rn})$ or $\PrimaryC{\Cluster_i}$ will be replaced via a local view-change. Eventually, these local view-changes will lead to a non-faulty primary in $\Cluster_i$, after which Proposition~\ref{prop:nc_send} again proves the statement with respect to $(\SignMessage{\Client_i}{\Transaction_i}, \Cert{\Cluster_i}{\SignMessage{\Client_i}{\Transaction_i}, \rn})$.
\end{proof}

\subsection{Ordering and Execution}\label{ss:order_exec}
Once replicas of a cluster have chosen a client request for execution and have received all client requests chosen by other clusters, they are ready for the final step: ordering and executing these client requests. 
In specific, in round $\rn$, any non-faulty replica that has valid requests from all clusters can 
move ahead and execute these requests.

Theorem~\ref{thm:learning_z} guarantees after the local replication step (Section~\ref{ss:local_rep}) 
and the inter-cluster sharing step (Section~\ref{ss:intershare}) each replica in $\Replicas$ will receive 
the same set of $\z$ client requests in round $\rn$. 
Let $S_{\rn} = \{ (\SignMessage{\Client_i}{\Transaction_i} \mid (1 \leq i \leq \z) \land (\Client_i \in \Clients{\Cluster_i}) \}$ be this set of $\z$ client requests received by each replica.

The last step is to put these client requests in a unique order, execute them, and 
inform the clients of the outcome. 
To do so, \GeoBFT{} simply uses a pre-defined ordering on the clusters. 
For example, each replica executes the transactions in the order $[\Transaction_1, \dots, \Transaction_{\z}]$. 
Once the execution is complete, each replica $\Replica \in \Cluster_i$, $1 \leq i \leq \z$, 
informs the client $\Client_i$ of any outcome (e.g., confirmation of execution or the result of execution).  
Note that each replica $\Replica$ only informs its local clients.
As all non-faulty replicas are expected to act deterministic, execution will yield the same state and results across all non-faulty replicas.
Hence, each client $\Client_i$ is guaranteed to receive identical response from at least $\f+1$ replicas. As there are at most $\f$ faulty replicas per cluster and faulty replicas cannot impersonate non-faulty replicas, at least one of these $\f+1$ responses must come from a non-faulty replica. We conclude the following:

\begin{theorem}[\GeoBFT{} is a consensus protocol]\label{thm:geobft}
Let $\Clusters$ be a system over $\Replicas$ in which every cluster satisfies $\n > 3\f$. A single round of \GeoBFT{} satisfies the following two requirements: 
\begin{description}[nosep]
    \item[Termination] If communication is reliable and has bounded delay, then \GeoBFT{} guarantees that each non-faulty replica in $\Replicas$ executes $\z$ transactions.
    \item[Non-divergence] \GeoBFT{} guarantees that all non-faulty replicas execute the same $\z$ transaction.
\end{description}
\end{theorem}
\begin{proof}
Both \emph{termination} and \emph{non-divergence} are direct corollaries
of Theorem~\ref{thm:learning_z}.
\end{proof}

\subsection{Final Remarks}
Until now we have presented the design of \GeoBFT{} using a strict notion of rounds. 
Only during the last step of each round of \GeoBFT{}, which orders and executes client requests (Section~\ref{ss:order_exec}), this strict notion of rounds is required.
All other steps can be performed out-of-order. 
For example, local replication and inter-cluster sharing of client requests for future rounds 
can happen in parallel with ordering and execution of client requests. 
In specific, the replicas of a cluster $\Cluster_i$, $1 \leq i \leq z$ 
can replicate the requests for round $\rn+2$, 
share the requests for round $\rn+1$ with other clusters, and 
execute requests for round $\rn$ in parallel.
Hence, \GeoBFT{} needs minimal synchronization between clusters. 

Additionally, we do not require that every cluster always has client requests available. 
When a cluster $\Cluster$ does not have client requests to execute in a round, 
the primary $\PrimaryC{\Cluster}$ can propose a no-op-request. The primary $\PrimaryC{\Cluster}$ can detect the need for such a no-op request in round $\rn$ when it starts receiving client requests for round $\rn$ from other clusters. As with all requests, also such no-op requests requires commit certificates obtained via local replication. 

To prevent that $\PrimaryC{\Cluster}$ can indefinitely ignore requests from some or all clients in $\Clients{\Cluster}$, we rely on standard \PBFT{} techniques to detect and resolve such attacks during local replication. These techniques effectively allow clients in $\Clients{\Cluster}$ to force the cluster to process its request, ruling out the ability of faulty primaries to indefinitely propose no-op requests when client requests are available.

Furthermore, to simplify presentation, we have assumed that every cluster has exactly the same size and that the set of replicas never change. These assumptions can be lifted, however. \GeoBFT{} can easily be extended to also work with clusters of varying size, this only requires minor tweaks on the remote view-change protocol of Figure~\ref{fig:rvc} (the conditions at Line~\ref{fig:rvc:trigger} rely on the cluster sizes, see Proposition~\ref{prop:remote_vc}). To deal with faulty replicas that eventually recover, we can rely on the same techniques as \PBFT{}~\cite{pbft,pbftj}. Full dynamic membership, in which replicas can join and leave \GeoBFT{} via some vetted automatic procedure, is a challenge for any permissioned blockchain and remains an open problem for future work~\cite{groupcom,gmp}.

\section[Implementation in ResilientDB]{Implementation in R\MakeLowercase{esilient}DB}

\GeoBFT{} is designed to enable geo-scale deployment of a permissioned blockchain. 
Next, we present our \ResilientDB{} fabric~\cite{resilientdb}, a permissioned blockchain fabric that
uses \GeoBFT{} to provide such a geo-scale aware high-performance permissioned blockchain. 
\ResilientDB{} is especially tuned to enterprise-level blockchains in which 
\begin{enumerate*}[label=(\roman*)]
\item replicas can be dispersed over a wide area network; 
\item links connecting replicas at large distances have low bandwidth; 
\item replicas are untrusted but known; and  
\item applications require high throughput and low latency.
\end{enumerate*}
These four properties are directly motivated by practical properties of geo-scale 
deployed distributed systems (see Table~\ref{tbl:geoscale_cost} in Section~\ref{sec:intro}). In Figure~\ref{fig:rdb_architect}, we present the architecture of \ResilientDB{}. 

\begin{figure}[t]
    \centering
    \includegraphics[width=\columnwidth]{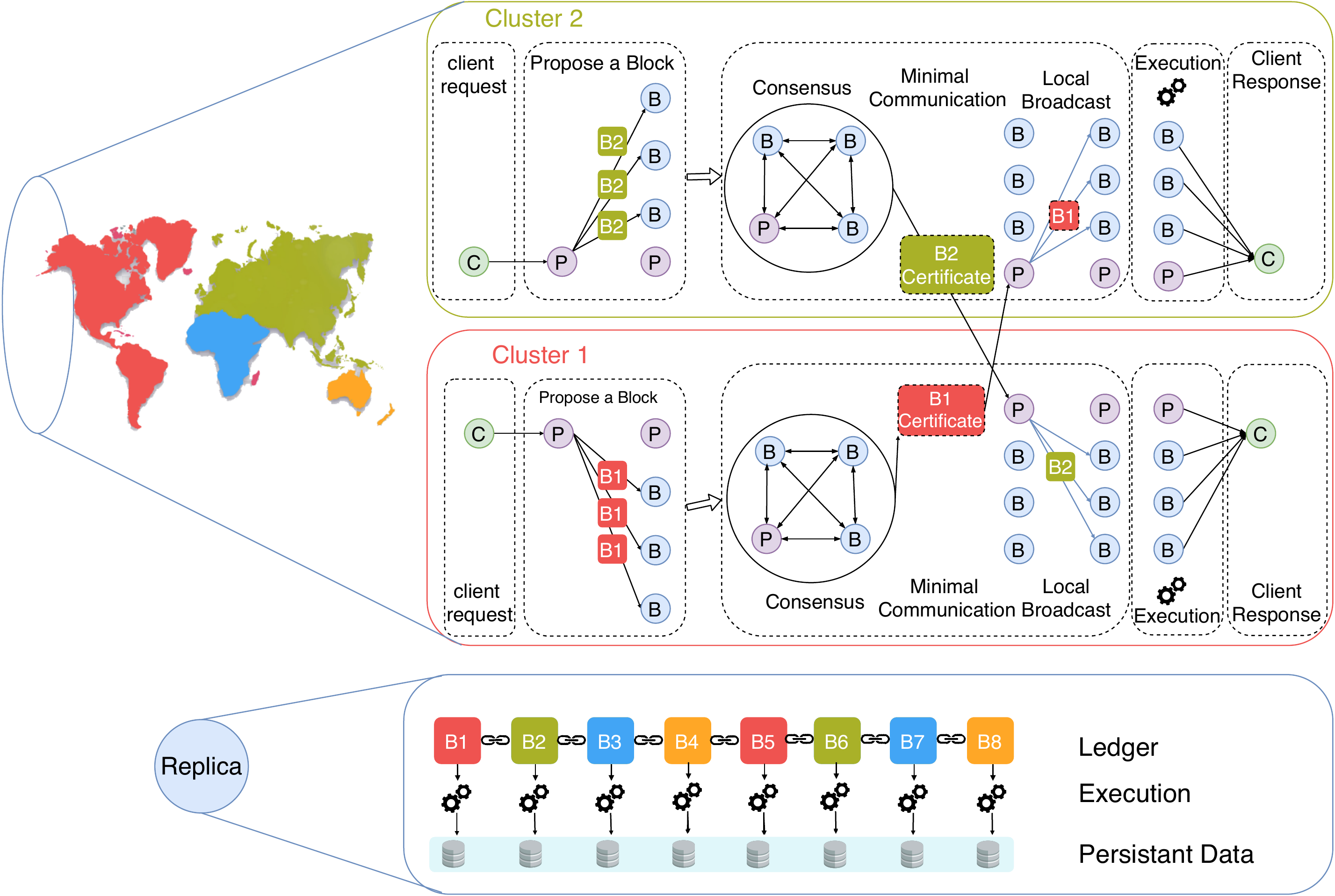}
    \caption{Architecture of our \ResilientDB{} Fabric.}
    \label{fig:rdb_architect}
    \vspace{-4mm}
\end{figure}

At the core of \ResilientDB{} is the ordering of client requests and appending them to the \emph{ledger}---the 
immutable append-only blockchain representing the ordered sequence of accepted client requests. 
The order of each client request will be determined via \GeoBFT{}, which we described in Section~\ref{sec:geobft}. 
Next, we focus on how the ledger is implemented and on other practical details that enable geo-scale performance.

\paragraph*{The ledger (blockchain)}
The key purpose of any blockchain fabric is the maintenance of the ledger: 
the immutable append-only blockchain representing the ordered sequence of client requests accepted. 
In \ResilientDB{}, the $i$-th block in the ledger consists of the $i$-th executed client request. 
Recall that in each round $\rn$ of \GeoBFT{}, each replica executes $\z$ requests, 
each belonging to a different cluster $\Cluster_i$, $1 \leq i \leq \z$.
Hence, in each round $\rn$, each replica creates $\z$ blocks in the order of execution of the $\z$ requests.
To assure immutability of each block, the block not only consists of the client request, 
but also contains a commit certificate. 
This prevents tampering of any block, as only a single commit certificate can be made per cluster 
per \GeoBFT{} round (Lemma~\ref{lem:pbft}). 
As \ResilientDB{} is designed to be a fully-replicated blockchain, each replica independently maintains 
a full copy of the ledger. 
The immutable structure of the ledger also helps when recovering replicas: 
tampering of its ledger by any replica can easily be detected. 
Hence, a recovering replica can simply read the ledger of any replica it chooses and directly verify whether the ledger can be trusted (is not tampered with).

\paragraph*{Cryptography}
The implementation of \GeoBFT{} and the ledger requires the availability of strong  cryptographic primitives, e.g., to provide digital signatures and authenticated communication (see Section~\ref{ss:prelim}). To do so, \ResilientDB{} provides its replicas and clients access to NIST-recommended strong cryptographic primitives~\cite{nist}. In specific, we use ED25519-based digital signatures to sign our messages and we use AES-CMAC message authentication codes to implement authenticated communication~\cite{cryptobook}. Further, we employ SHA256 to generate collision-resistant message digests.

\paragraph*{Pipelined consensus}
From our experience designing and implementing Byzantine consensus protocols, we know that throughput can be limited by waiting (e.g., due to message latencies) or by computational costs (e.g., costs of signing and verifying messages). To address both issues simultaneously, \ResilientDB{} provides a multi-threaded pipelined architecture for the implementation of consensus protocols. In Figure~\ref{fig:geobft_pipe}, we have illustrated how \GeoBFT{} is implemented in this multi-threaded pipelined architecture.

\begin{figure*}[t]
    \vspace{-1mm}
    \centering
    \includegraphics[width=0.44\textwidth]{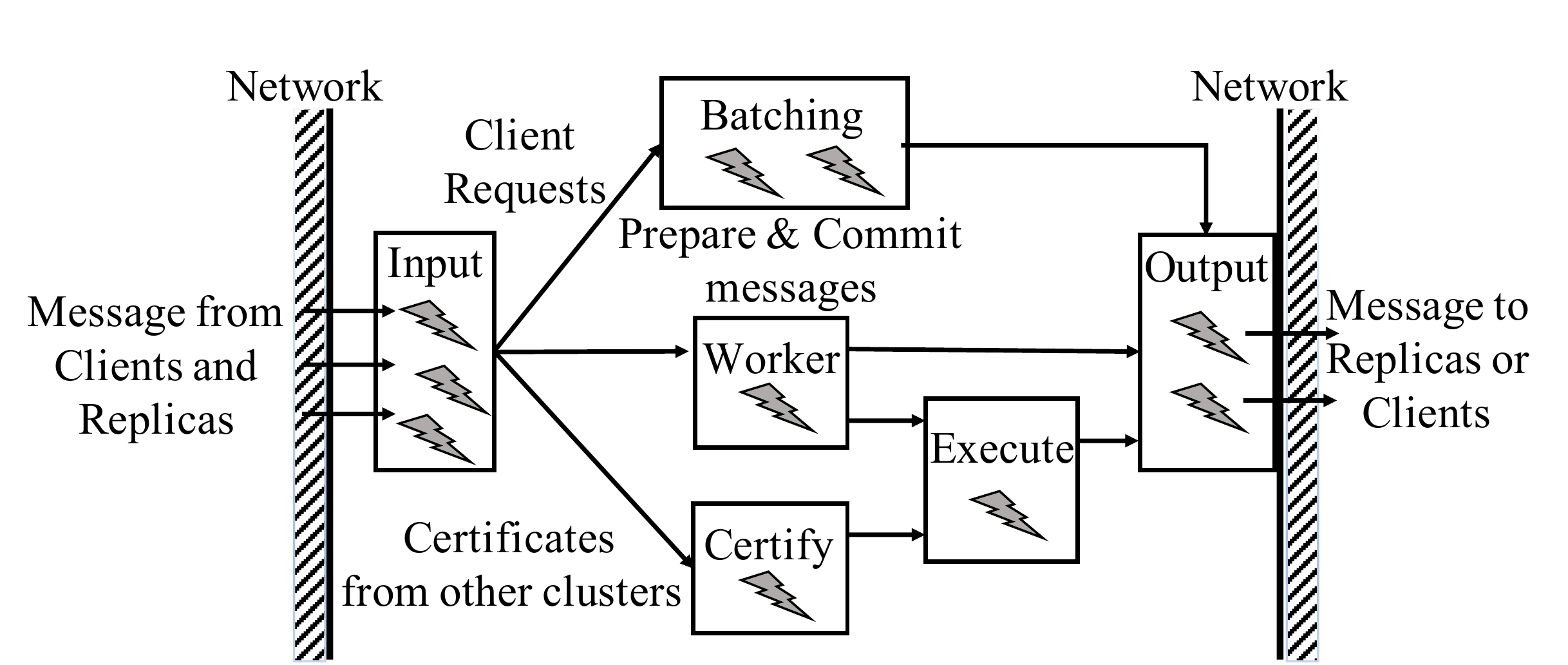}
    \hspace{1.5cm}
    \includegraphics[width=0.44\textwidth]{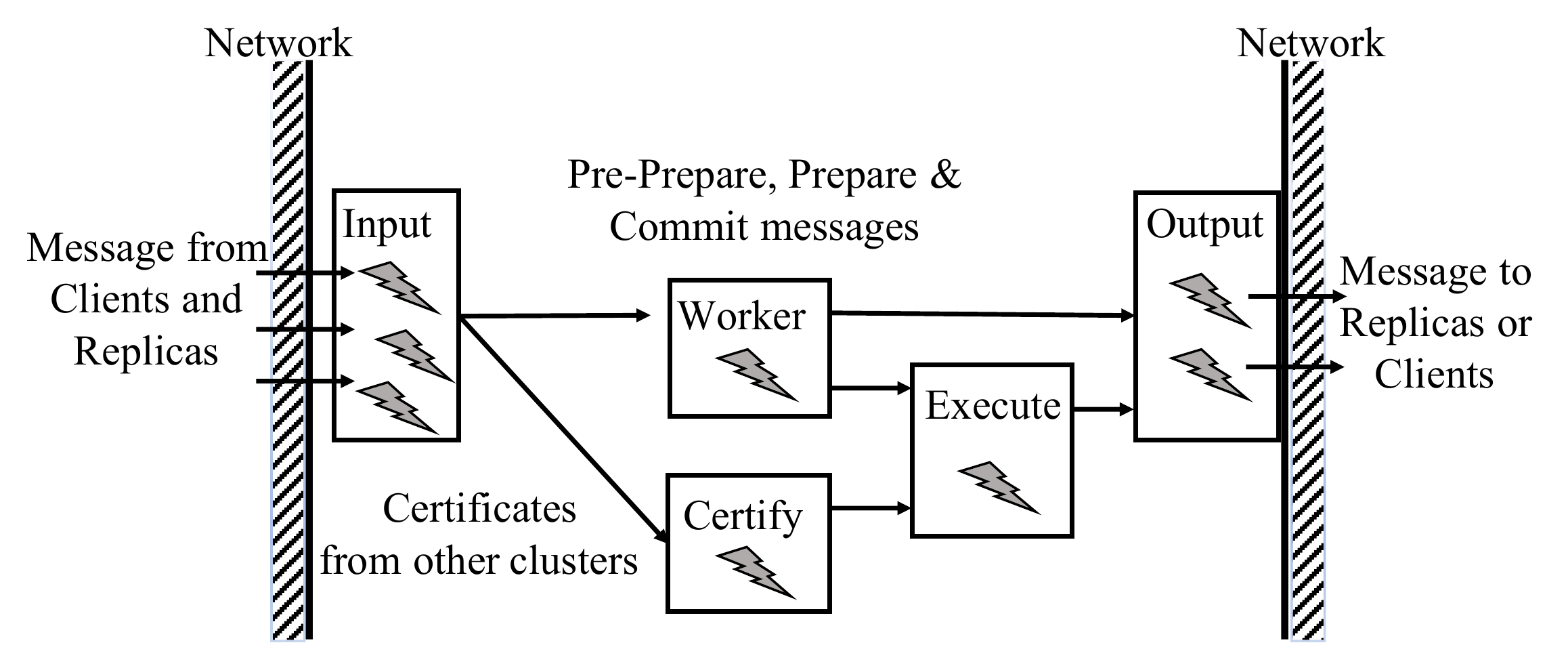}
    \vspace{-2mm}
    \caption{The multi-threaded implementation of \GeoBFT{} in \ResilientDB{}. \emph{Left}, the implementation for local primaries. \emph{Right}, the implementation for other replicas.
    }\label{fig:geobft_pipe}
\end{figure*}

With each replica, we associate a set of \emph{input threads} that receive messages from the network.  The primary has one input thread dedicated to accepting client requests, which this input thread places on the batch queue for further processing. All replicas have two input threads for processing all other messages (e.g., those related to local replication and global sharing). Each replica has two \emph{output threads} for sending messages. The ordering (consensus) and execution of client requests is done by the \emph{worker}, \emph{execute}, and \emph{certify} threads. 
We provide details on these threads next.

\paragraph*{Request batching}
In the design of \ResilientDB{} and \GeoBFT{}, 
we support the grouping of client requests in \emph{batches}. 
Clients can group their requests in batches and send these batches to their local cluster.
Furthermore, local primaries can group requests of different clients into a single batch. 
Each batch is then processed by the consensus protocol (\GeoBFT{}) as a single request, 
thereby sharing the cost associated with reaching consensus among each of the requests in the batch. Such request batching is a common practice in high-performance blockchain systems and can be applied to a wide range of workloads (e.g., processing financial transactions).

When an input thread at a local primary $\PrimaryC{\Cluster}$ receives a batch of client requests from a client, the thread assigns it a linearly increasing number and places the batch in the \emph{batch queue}. Via this assigned number, $\PrimaryC{\Cluster}$ already decided the order in which this batch needs to be processed (and the following consensus steps are only necessary to communicate this order reliably with all other replicas). Next, the \emph{batching thread} at $\PrimaryC{\Cluster}$ takes request batches from the batch queue and initiates local replication. 
To do so, the batching thread initializes the data-structures used by the local replication step, 
creates a valid \Name{pre-prepare} message for the batch (see Section~\ref{ss:local_rep}), and puts this message on the output queue. 
The output threads dequeue these messages from the output queue and send them to all intended recipients.

When a input thread at replica $\Replica \in \Cluster$ receives a message as 
part of local replication (e.g., \Name{pre-prepare}, \Name{prepare}, \Name{commit}, 
see Section~\ref{ss:local_rep}), then it places this message in the \emph{work queue}. 
Next, the \emph{worker thread} processes any incoming messages placed on the work queue 
by performing the steps of \PBFT{}, the local replication protocol. 
At the end of the local replication step, when the worked thread has received 
$\n-\f$ identical \Name{commit} messages for a batch, it notifies both the 
\emph{execution thread} and \emph{certify thread}.
The certify thread creates a commit certificate corresponding to this notification 
and places this commit certificate in the output queue, initiating inter-cluster sharing
(see Section~\ref{ss:intershare}).

When a input thread at replica $\Replica \in \Cluster$ receives a message as 
part of global sharing (a commit certificate), then it places this message in the \emph{certify queue}. 
Next, the certify thread processes any incoming messages placed on the certify queue 
by performing the steps of the global sharing protocol. 
If the certify thread has processed commit certificates for batches in round $\rn$ from all clusters, 
then it notifies the execution thread that round $\rn$ can be ordered and executed.

The execute thread at replica $\Replica \in \Cluster$ waits until it receives notification for all commit certificates associated with the next round. These commit certificates correspond to all the client batches that need to be ordered and executed. When the notifications are received, the execute thread simply performs the steps described in Section~\ref{ss:order_exec}.

\paragraph*{Other protocols} 
\ResilientDB{} also provides implementations of \emph{four} other state-of-the-art consensus protocols, namely, \PBFT{}, \ZZ{}, \HS{}, and \STW{} (see Section~\ref{ss:challenges}). Each of these protocols use the multi-threaded pipelined architecture of \ResilientDB{} and are structured similar to the design of \GeoBFT{}. Next, we provide some details on each of these protocols.

We already covered the working of \PBFT{} in Section~\ref{ss:local_rep}.  Next, \ZZ{} is designed with the most optimal case in mind: it requires non-faulty clients and depends on clients to aid in the recovery of \emph{any} failures. To do so, clients in \ZZ{} require identical responses from all $\n$ replicas. If these are not received, the client initiates recovery of any requests with sufficient $\n-\f$ responses by broadcasting certificates of these requests. This will greatly reduce performance when \emph{any} replicas are faulty.  In \ResilientDB{}, the certify thread at each replica processes these recovery certificates.

\HS{} is designed to reduce communication of \PBFT{}. To do so, \HS{} uses threshold signatures to combine $\n-\f$ message signatures into a single signature. As there is no readily available implementation for threshold signatures available for the Crypto++ library, we skip the construction and verification of threshold signatures in our implementation of \HS{}. Moreover, we allow each replica of \HS{} to act as a primary in parallel without requiring the usage of pacemaker-based synchronization~\cite{hotstuff}. Both decisions give our \HS{} implementation a substantial performance advantage in practice.

Finally, we also implemented \STW{}. This protocol groups replicas into clusters, similar to \GeoBFT{}. Different from \GeoBFT{}, \STW{} designates one of these clusters as the \emph{primary cluster}, which coordinates all operations. To reduce inter-cluster communication, \STW{} uses threshold signatures. As with \HS{}, we omitted these threshold signatures in our implementation.

\section{Evaluation}\label{sec:eval}

To showcase the practical value of \GeoBFT{}, we now use our \ResilientDB{} fabric 
to evaluate \GeoBFT{} against four other popular state-of-the-art consensus protocols 
(\PBFT{}, \ZZ, \HS{}, and \STW{}). 
We deploy \ResilientDB{} on the Google Cloud using \texttt{N1} machines that have $8$-core Intel Skylake CPUs and $\SI{16}{\giga\byte}$ of main memory. 
Additionally, we deploy $\SI{160}{\kilo\nothing}$ clients on eight $4$-core machines having 
$\SI{16}{\giga\byte}$ of main memory.
We equally distribute the clients across all the regions used in each experiment.

In each experiment, the workload is provided by the \emph{Yahoo Cloud Serving Benchmark}  (YCSB)~\cite{ycsb}. Each client transaction queries a YCSB table with an active set of $\SI{600}{\kilo\nothing}$ records. For our evaluation, we use \emph{write queries}, as those are typically more costly than read-only queries. Prior to the experiments, each replica is initialized with an identical copy of the YCSB table. The client transactions generated by YCSB follow a uniform Zipfian distribution. Clients and replicas can batch transactions to reduce the cost of consensus. In our experiments, we use a \emph{batch size} of $100$ requests per batch (unless stated otherwise). 

With a batch size of $100$, the messages have sizes of $\SI{5.4}{\kilo\byte}$ (\Name{preprepare}), $\SI{6.4}{\kilo\byte}$ (commit certificates containing seven \Name{commit} messages and a \Name{preprepare} message), $\SI{1.5}{\kilo\byte}$ (client responses), and $\SI{250}{\byte}$ (other messages).
The size of a commit certificate is largely dependent on the size of 
the \Name{preprepare} message, while the total size of the accompanying \Name{commit} 
messages is small. Hence, the inter-cluster sharing of these certificates is not a bottleneck for \GeoBFT{}: existing \BFT{} 
protocols send \Name{preprepare} messages to all replicas irrespective of their region. 
Further, if the size of \Name{commit} messages starts dominating, then 
threshold signatures can be adopted to reduce their cost~\cite{rsasign}.

To perform geo-scale experiments, we deploy replicas across \emph{six} 
different regions, namely Oregon, Iowa, Montreal, Belgium, Taiwan, and Sydney. 
In Table~\ref{tbl:geoscale_cost}, we present our measurements on the inter-region network latency and bandwidth. 
We run each experiment for $\SI{180}{\second}$: first, we allow the system to warm-up for $\SI{60}{\second}$, after which we collect measurement results for the next $\SI{120}{\second}$. We average the results of our experiments over three runs. 

For \PBFT{} and \ZZ{}, centralized protocols in which a single primary replica coordinates consensus, 
we placed the primary in Oregon, as this region has the highest bandwidth to all other regions 
(see Table~\ref{tbl:geoscale_cost}).
For \HS{}, our implementation permits all replicas to act as both primary and non-primary at the same time.
For both \GeoBFT{} and \STW{}, we group replicas in a single region into a single cluster. 
In each of these protocols, each cluster has its own local primary.
Finally, for \STW{}, a centralized protocol in which the primary cluster coordinates the consensus, 
we placed the primary cluster in Oregon.
We focus our evaluation on answering the following four research questions:
\begin{enumerate}[wide,nosep,label=(\arabic*)]
\item What is the impact of geo-scale deployment of replicas in distant clusters on the performance of \GeoBFT{}, as compared to other consensus protocols?
\item What is the impact of the size of local clusters (relative to the number of clusters) on the performance of \GeoBFT{}, as compared to other consensus protocols?
\item What is the impact of failures on the performance of \GeoBFT{}, as compared to other consensus protocols?
\item Finally, what is the impact of request batching on the performance of \GeoBFT{}, as compared to other consensus protocols, and under which batch sizes can \GeoBFT{} already provide good throughput?
\end{enumerate}

\pgfplotstableread{
clusters	pbft	geobft	hs	stw	zz
1	34342	9900	44876	17641	37618
2	33641	38322	43081	19618	35357
3	30174	41769	40272	19038	34503
4	17449	50143	40499	18717	18008
5	17362	54881	43089	18661	18070
6	14584	54936	42517	15556	15180
}\dataTPUTfnCLUSTER

\pgfplotstableread{
clusters	pbft	geobft	hs	stw	zz
1	0.666	0.694	5.317	0.960	0.437
2	0.512	0.840	6.077	0.891	0.454
3	0.541	0.645	5.909	0.867	0.550
4	1.280	0.735	7.261	1.192	1.242
5	1.807	0.880	7.124	1.123	1.649
6	2.066	0.906	7.607	1.617	1.994
}\dataLATfnCLUSTER

\pgfplotstableread{
replicas	pbft	geobft	hs	stw	zz
4	17644	100922	50847	18815	18004
7	17241	79980	48842	18590	17842
10	17403	62704	42034	18388	17991
12	17399	53992	40039	19022	17993
15	17449	50143	40499	18717	18008
}\dataTPUTfnREPLICAS

\pgfplotstableread{
replicas	pbft	geobft	hs	stw	zz
4	0.633	0.509	2.906	0.987	0.617
7	0.649	0.647	3.034	0.999	0.623
10	0.644	0.644	3.532	1.013	0.827
12	0.644	0.612	5.709	0.977	0.619
15	1.28	0.944	7.261	1.192	1.242
}\dataLATfnREPLICAS

\pgfplotstableread{
replicas	pbft	geobft	hs	stw	zz
4	17800	100393	45338	18672	261
7	17186	79159	44978	18297	301
10	16657	64064	43504	18978	261
12	16717	55861	43414	17766	301
}\dataTPUTfnREPLICASwFAIL

\pgfplotstableread{
replicas	pbft	geobft	hs	stw	zz
4	16675	100498	45081	17334	393
7	16758	76724	43280	17603	393
10	16657	60802	42407	18438	393
12	16717	52916	37190	17766	393
}\dataTPUTfnREPLICASwFFAIL

\pgfplotstableread{
replicas	pbft	geobft
4	17509	89302
7	17544	75082
10	17666	65006
12	15986	53621
}\dataTPUTfnREPLICASwVC

\pgfplotstableread{
batchsize	pbft	geobft	hs	stw	zz
10	9511	12327	8104	7866	12243
50	16251	50608	31433	17874	17342
100	17241	79980	48842	18590	17842
200	17916	99763	69259	19118	18579
300	18755	112319	79183	19063	18755
}\dataTPUTfnBATCH

\begin{figure*}[t]
    \centering
    \scalebox{0.7}{\ref*{legend:all_line}}\\[10pt]
    \begin{minipage}{0.49\textwidth}
        \centering
        \begin{tikzpicture}[geobftplot]
            \begin{axis}[title={Throughput},
                         xlabel={Number of clusters},
                         ylabel={\smash{Throughput (\si{\text{txn}\per\second})}},
                         legend to name=legend:all_line,
                         legend columns=-1]
                \addplot table[x={clusters},y={geobft}] {\dataTPUTfnCLUSTER};
                \addplot table[x={clusters},y={pbft}] {\dataTPUTfnCLUSTER};
                \addplot table[x={clusters},y={zz}] {\dataTPUTfnCLUSTER};
                \addplot table[x={clusters},y={hs}] {\dataTPUTfnCLUSTER};
                \addplot table[x={clusters},y={stw}] {\dataTPUTfnCLUSTER};
                \legend{\GeoBFT{},\PBFT{},\ZZ{},\HS{},\STW{}};
            \end{axis}
        \end{tikzpicture}\quad
        \begin{tikzpicture}[geobftplot]
            \begin{axis}[title={Latency},
                         xlabel={Number of clusters},
                         ylabel={\smash{Latency (\si{\second})}}]
                \addplot table[x={clusters},y={geobft}] {\dataLATfnCLUSTER};
                \addplot table[x={clusters},y={pbft}] {\dataLATfnCLUSTER};
                \addplot table[x={clusters},y={zz}] {\dataLATfnCLUSTER};
                \addplot table[x={clusters},y={hs}] {\dataLATfnCLUSTER};
                \addplot table[x={clusters},y={stw}] {\dataLATfnCLUSTER};
            \end{axis}
        \end{tikzpicture}
    \end{minipage}
    \hfill
    \begin{minipage}{0.49\textwidth}
        \centering
        \begin{tikzpicture}[geobftplot]
            \begin{axis}[title={Throughput},
                         xlabel={Number of replicas (per cluster)},
                         ylabel={\smash{Throughput (\si{\text{txn}\per\second})}},
                         xtick={4,7,10,12,15},
                         scaled y ticks=base 10:-4]
                \addplot table[x={replicas},y={geobft}] {\dataTPUTfnREPLICAS};
                \addplot table[x={replicas},y={pbft}] {\dataTPUTfnREPLICAS};
                \addplot table[x={replicas},y={zz}] {\dataTPUTfnREPLICAS};
                \addplot table[x={replicas},y={hs}] {\dataTPUTfnREPLICAS};
                \addplot table[x={replicas},y={stw}] {\dataTPUTfnREPLICAS};
            \end{axis}
        \end{tikzpicture}\quad
        \begin{tikzpicture}[geobftplot]
            \begin{axis}[title={Latency},
                         xlabel={Number of replicas (per cluster)},
                         ylabel={\smash{Latency (\si{\second})}},
                         xtick={4,7,10,12,15}]
                \addplot table[x={replicas},y={geobft}] {\dataLATfnREPLICAS};
                \addplot table[x={replicas},y={pbft}] {\dataLATfnREPLICAS};
                \addplot table[x={replicas},y={zz}] {\dataLATfnREPLICAS};
                \addplot table[x={replicas},y={hs}] {\dataLATfnREPLICAS};
                \addplot table[x={replicas},y={stw}] {\dataLATfnREPLICAS};
            \end{axis}
        \end{tikzpicture}
    \end{minipage}
    \begin{minipage}[t]{0.49\textwidth}
        \caption{Throughput and latency as a function of the number of clusters; $\z\n = 60$ replicas.}\label{fig:fn_clust}
    \end{minipage}
    \hfill
    \begin{minipage}[t]{0.49\textwidth}
        \caption{Throughput and latency as a function of the number of replicas per cluster; $\z = 4$.}\label{fig:fn_repl}
    \end{minipage}\\[10pt]
    \begin{minipage}{0.7425\textwidth}
        \centering
        \begin{tikzpicture}[geobftplot]
            \begin{axis}[title={Throughput (one failure)},
                         xlabel={Number of replicas (per cluster)},
                         ylabel={\smash{Throughput (\si{\text{txn}\per\second})}},
                         xtick={4,7,10,12},
                         scaled y ticks=base 10:-4]
                \addplot table[x={replicas},y={geobft}] {\dataTPUTfnREPLICASwFAIL};
                \addplot table[x={replicas},y={pbft}] {\dataTPUTfnREPLICASwFAIL};
                \addplot table[x={replicas},y={zz}] {\dataTPUTfnREPLICASwFAIL};
                \addplot table[x={replicas},y={hs}] {\dataTPUTfnREPLICASwFAIL};
                \addplot table[x={replicas},y={stw}] {\dataTPUTfnREPLICASwFAIL};
            \end{axis}
        \end{tikzpicture}\quad
        \begin{tikzpicture}[geobftplot]
            \begin{axis}[title={Throughput ($\f$ failures)},
                         xlabel={Number of replicas (per cluster)},
                         ylabel={\smash{Throughput (\si{\text{txn}\per\second})}},
                         xtick={4,7,10,12},
                         scaled y ticks=base 10:-4]
                \addplot table[x={replicas},y={geobft}] {\dataTPUTfnREPLICASwFFAIL};
                \addplot table[x={replicas},y={pbft}] {\dataTPUTfnREPLICASwFFAIL};
                \addplot table[x={replicas},y={zz}] {\dataTPUTfnREPLICASwFFAIL};
                \addplot table[x={replicas},y={hs}] {\dataTPUTfnREPLICASwFFAIL};
                \addplot table[x={replicas},y={stw}] {\dataTPUTfnREPLICASwFFAIL};
            \end{axis}
        \end{tikzpicture}\quad
        \begin{tikzpicture}[geobftplot]
            \begin{axis}[title={Throughput (primary failure)},
                         xlabel={Number of replicas (per cluster)},
                         ylabel={\smash{Throughput (\si{\text{txn}\per\second})}},
                         xtick={4,7,10,12}]
                \addplot table[x={replicas},y={geobft}] {\dataTPUTfnREPLICASwVC};
                \addplot table[x={replicas},y={pbft}] {\dataTPUTfnREPLICASwVC};
            \end{axis}
        \end{tikzpicture}
    \end{minipage}
    \hfill
    \begin{minipage}{0.2375\textwidth}
        \centering
        \begin{tikzpicture}[geobftplot]
            \begin{axis}[title={Throughput},
                         xlabel={Transactions per batch (batch size)},
                         ylabel={\smash{Throughput (\si{\text{txn}\per\second})}},
                         xtick={10,50,100,200,300},
                         scaled y ticks=base 10:-4]
                \addplot table[x={batchsize},y={geobft}] {\dataTPUTfnBATCH};
                \addplot table[x={batchsize},y={pbft}] {\dataTPUTfnBATCH};
                \addplot table[x={batchsize},y={zz}] {\dataTPUTfnBATCH};
                \addplot table[x={batchsize},y={hs}] {\dataTPUTfnBATCH};
                \addplot table[x={batchsize},y={stw}] {\dataTPUTfnBATCH};
            \end{axis}
        \end{tikzpicture}
    \end{minipage}
    \begin{minipage}[t]{0.7425\textwidth}
        \caption{Throughput as a function of the number of replicas per cluster; $\z = 4$. \emph{Left}, throughput with one non-primary failure. \emph{Middle}, throughput with $\f$ non-primary failures. \emph{Right}, throughput with a single primary failure.}\label{fig:fn_fail}
    \end{minipage}
    \hfill
    \begin{minipage}[t]{0.2375\textwidth}
        \caption{Throughput as a function of the batch size; $\z = 4$ and $\n = 7$.}\label{fig:fn_batch}
    \end{minipage}
\end{figure*}

\subsection{Impact of Geo-Scale deployment}

First, we determine the impact of geo-scale deployment of replicas in distant regions on the performance of \GeoBFT{} and other consensus protocols. To do so, we measure the throughput and latency attained by \ResilientDB{} as a function of the number of regions, which we vary between $1$ and $6$. 
We use $60$ replicas evenly distributed over the regions, and we select regions in the order Oregon, Iowa, Montreal, Belgium, Taiwan, and Sydney. E.g., if we have four regions, then each region has $15$ replicas, and 
we have these replicas in Oregon, Iowa, Montreal, and Belgium. 
The results of our measurements can be found in Figure~\ref{fig:fn_clust}.

From the measurements, we see that \STW{} is unable to benefit from its topological knowledge of the network: in practice, we see that the high computational costs and the centralized design of \STW{} prevent high throughput in all cases. 
Both \PBFT{} and \ZZ{} perform better than \STW{}, especially when ran in a few well-connected regions (e.g., only the North-American regions). 
The performance of these protocols falls when inter-cluster communication becomes a bottleneck, however (e.g., when regions are spread across continents). 
\HS{}, which is designed to reduce communication compared to \PBFT{}, has reasonable  throughput in a geo-scale deployment, and sees only a small drop in throughput when regions are added.  
The high computational costs of the protocol prevent it from reaching high throughput in any setting, however. 
Additionally, \HS{} has very high latencies due to its $4$-phase design. 
As evident from Figure~\ref{tbl:compare}, \HS{} clients face severe delay in receiving a response for their client requests.

Finally, the results clearly show that \GeoBFT{} scales well with an increase in regions. 
When running at a single cluster, the added overhead of \GeoBFT{} (as compared to \PBFT{}) is high, 
which limits its throughput in this case. 
Fortunately, \GeoBFT{} is the only protocol that actively benefits from adding regions: 
adding regions implies adding clusters, which \GeoBFT{} uses to increase parallelism of consensus and 
decrease centralized communication. 
This added parallelism helps offset the costs of inter-cluster communication, 
even when remote regions are added. 
Similarly, adding remote regions only incurs a low latency on \GeoBFT{}.
Recall that \GeoBFT{} sends only $\f+1$ messages between any two clusters. Hence, a total of $\BigO{\z\f}$ inter-cluster messages are sent, which is much less than the number of messages communicated across clusters by other protocols (see Figure~\ref{tbl:compare}). As the cost of communication between remote clusters is high (see Figure~\ref{tbl:geoscale_cost}), this explains why other protocols have lower throughput and higher latencies than \GeoBFT{}. Indeed, when operating on several regions, \GeoBFT{} is able to outperform \PBFT{} by a factor of up-to-$3.1\times$ and outperform \HS{} by a factor of up-to-$1.3\times$.

\subsection{Impact of Local Cluster Size}

Next, we determine the impact of the number of replicas per region on the performance of \GeoBFT{} and other consensus protocols. To do so, we measure the throughput and latency attained by \ResilientDB{} as a function of the number of replicas per region, which we vary between $4$ and $15$. We have replicas in four regions (Oregon, Iowa, Montreal, and Belgium). The results of our measurements can be found in Figure~\ref{fig:fn_repl}.

The measurements show that increasing the number of replicas only has minimal negative influence on the throughput and latency of \PBFT{}, \ZZ{}, and \STW{}. As seen in the previous Section, the inter-cluster communication cost for the primary to contact individual replicas in other regions (and continents) is the main bottleneck. Consequently, the number of replicas used only has minimal influence. For \HS{}, which does not have such a bottleneck, adding replicas does affect throughput and---especially---latency, this due to the strong dependence between latency and the number of replicas in the design of \HS{}. 

The design of \GeoBFT{} is particularly tuned toward a large number of regions (clusters), and not toward a large number of replicas per region. We observe that increasing the replicas per cluster also allows each cluster to tolerate more failures (increasing $\f$). Due to this, the performance drop-off for \GeoBFT{} when increasing the replicas per region is twofold: first, the size of the certificates  exchanged between clusters is a function of $\f$; second, each cluster sends their certificates to $\f  + 1$ replicas in each other cluster.  Still, the parallelism incurred by running in four clusters allows \GeoBFT{} to outperform all other protocols, even when scaling up to fifteen replicas per region, in which case it is still $2.9\times$ faster than \PBFT{} and $1.2\times$ faster than \HS{}.

\subsection{Impact of Failures}

In our third experiment, we determine the impact of replica failures on the performance of \GeoBFT{} and other consensus protocols. To do so, we measure the throughput attained by \ResilientDB{} as a function of the number of replicas, which we vary between $4$ and $12$. We perform the measurements under three failure scenarios: a single non-primary replica failure, up to $\f$ simultaneous non-primary replica failures per region, and a single primary failure. As in the previous experiment, we have replicas in four regions (Oregon, Iowa, Montreal, and Belgium). The results of our measurements can be found in Figure~\ref{fig:fn_fail}.

\paragraph*{Single non-primary replica failure}
The measurements for this case show that the failure of a single non-primary replica has only a small impact on the throughput of most protocols. The only exception being \ZZ{}, for which the throughput plummets to zero, as \ZZ{} is optimized for the optimal non-failure case. The inability of \ZZ{} to effectively operate under any failures is consistent with prior analysis of the protocol~\cite{zzuns,aadvark}.

\paragraph*{$\f$ non-primary replica failures per cluster} In this experiment, we measure the performance of \GeoBFT{} in the worst case scenario it is designed for: the simultaneous failure of $\f$ replicas in each cluster ($\f\z$ replicas in total). This is also the worst case \STW{} can deal with, and is close to the worst case the other protocols can deal with (see Remark~\ref{rem:model}).

The measurements show that the failures have a moderate impact on the performance of all protocols (except for \ZZ{} which, as in the single-failure case, sees its throughput plummet to zero). The reduction in throughput is a direct consequence of the inner working of the consensus protocols. Consider, e.g., \GeoBFT{}. In \GeoBFT{}, replicas in each cluster first choose a local client request and replicate this request locally using \PBFT{} (see Section~\ref{ss:local_rep}). In each such local replication step, each replica will have two phases in which it needs to receive $\n - \f$ identical messages before proceeding to the next phase (namely, \Name{prepare} and \Name{commit} messages). If there are no failures, then each replica only must wait for the $\n-\f$ fastest messages and can proceed to the next phase as soon as these messages are received (ignoring any delayed messages). However, if there are $\f$ failures, then each replica must wait for all messages of the remaining non-failed replicas to arrive before proceeding, including the slowest arriving messages. Consequently, the impact of temporary disturbances causing random message delays at individual replicas increases with the number of failed replicas, which negatively impacts performance. Similar arguments also hold for \PBFT{}, \STW{}, and \HS{}.

\paragraph*{Single primary failure} In this experiment, we measure the performance of \GeoBFT{} if a single primary fails (in one of the four regions). We compare the performance of \GeoBFT{} with \PBFT{} under failure of a single primary,  which will cause primary replacement via a view-change. 
For \PBFT{}, we require checkpoints to be generated and transmitted after every $600$ client transactions.
Further, we perform the primary failure after $900$ client transactions have been ordered.

For \GeoBFT{}, we fail the primary of the cluster in Oregon once each cluster has ordered $900$ transactions.
Similarly, each cluster exchanges checkpoints periodically, after locally replicating every $600$ transactions.
In this experiment, we have excluded \ZZ{}, as it already fails to deal with non-primary failures, \HS{}, as it utilizes rotating primaries and does not have a notion of a fixed primary, and \STW{}, as it does not provide a readily-usable and complete view-change implementation. As expected, the measurements show that recovery from failure incurs a small reduction in overall throughput in both protocols, as both protocols are able to recover to normal-case operations after failure.

\subsection{Impact of Request Batching}

We now determine the impact of the batch size---the number of client transactions processed by the consensus protocols in a single consensus decision---on the performance of various consensus protocols. To do so, we measure the throughput attained by \ResilientDB{} as a function of the batch size, which we vary between $10$ and $300$. For this experiment, we have replicas in four regions (Oregon, Iowa, Montreal, and Belgium), and each region has seven replicas. The results of our measurements can be found in Figure~\ref{fig:fn_batch}.

The measurements show a clear distinction between, on the one hand, \PBFT{}, \ZZ{}, and \STW{}, and, on the other hand, \GeoBFT{} and \HS{}. Note that in \PBFT{}, \ZZ{}, and \STW{} a single primary residing in a single region coordinates all consensus. This highly centralized communication limits throughput, as it is bottlenecked by the bandwidth of the single primary. 
\GeoBFT{}---which has primaries in each region---and \HS{}---which rotates primaries---both distribute consensus 
over several replicas in several regions, removing bottlenecks due to the bandwidth of any single replica. 
Hence, these protocols have sufficient bandwidth to support larger batch sizes (and increase throughput). 
Due to this, \GeoBFT{} is able to outperform \PBFT{} by up-to-$6.0\times$. 
Additionally, as the design of \GeoBFT{} is optimized to minimize global bandwidth usage, \GeoBFT{} is even able to outperform \HS{} by up-to-$1.6\times$.

\section{Related Work}

Resilient systems and consensus protocols have been widely studied by the distributed computing community (e.g.,~\cite{easyc,easyc-extend,icdt_delayed,pdbook,specpaxos,quecc,tbook,distalgo,distbook}). 
Here, we restrict ourselves to works addressing some of the challenges addressed by \GeoBFT{}: consensus protocols supporting high-performance or geo-scale aware resilient system designs.

\paragraph*{Traditional \BFT{} consensus} 
The consensus problem and related problems such as Byzantine Agreement and Interactive Consistency have been studied in detail since the late seventies~\cite{cryptorounds,netbound,dolevstrong,byzgen,dolevbound,interbound,flp,paxos,generals,generalmw,consbound,possibleasync,distalgo}. The introduction of the \PBFT{}-powered \emph{BFS}---a fault-tolerant version of the networked file system~\cite{nfs}---by Castro et al.~\cite{pbft,pbftj} marked the first practical high-performance system using consensus. Since the introduction of \PBFT{}, many consensus protocols have been proposed that improve on aspects of \PBFT{}, e.g, \ZZ{}, \textsc{PoE}, and \HS{}, as discussed in the Introduction. To further improve on the performance of \PBFT{}, some consensus protocols consider providing less failure resilience~\cite{qubft,qclients,hq,quorumbyzclients,byzq,phalanx}, or rely on trusted components~\cite{hybster,less-replica-2,less-replica-1,cheapbft,ebawa,less-replica-3}.  
Some recent protocols propose the notion of multiple parallel primaries~\cite{disc_mbft,multibft-system}. 
Although such designs have partial decentralization, these protocols are still not geo-scale aware.
None of these protocols are fully geo-scale aware, however, making them unsuitable for the setting we envision for \GeoBFT{}.

Protocols such as \STW{}, \Name{Blink},  and \Name{Menicus} improve on \PBFT{} by partly optimizing for geo-scale aware deployments~\cite{blink,steward,menicus,geo-uncivil}. In Section~\ref{sec:eval}, we already showed that the design of \STW{}---which depends on a primary cluster---severely limits its performance.  The \Name{Blink} protocol improved the design of \STW{} by removing the need for a primary cluster, as it requires each cluster to order all incoming requests~\cite{blink}. This design comes with high communication costs, unfortunately.  Finally, \Name{Menicus} tries to reduce communication costs for clients by letting clients only communicate with close-by replicas~\cite{menicus,geo-uncivil}. By alternating the primary among all replicas, this allows clients to propose requests without having to resort to contacting replicas at geographically large distances. As such, this design focusses on the costs perceived by clients, whereas \GeoBFT{} focusses on the overall costs of consensus.

\paragraph*{Sharded Blockchains} Sharding is an indispensable tool used by database systems to deal with Big Data~\cite{postgrexl,mysqlshard,oracleshard,msshard,pdbook,distbook}. 
Unsurprisingly, recent blockchain systems such as SharPer, Elastico, Monoxide, AHL, and RapidChain explore the use of sharding within the design of a replicated blockchain~\cite{caper,sharper,ahl,elastico,monoxide,rapidchain}. 
To further enable sharded designs, also high-performance communication primitives that enable communication between fault-tolerant clusters (shards) have been proposed~\cite{disc_csp}. 
Sharding does not fully solve performance issues associated with geo-scale deployments of permissioned blockchains, however. 
The main limitation of sharding is the efficient evaluation of complex operations across 
shards~\cite{bernstein,pdbook}, and sharded systems achieve high throughputs only if large portions of the 
workload access single shards. 
For example, in SharPer~\cite{sharper} and AHL~\cite{ahl}, two recent permissioned blockchain designs,
consensus on the cross-shard transactions is achieved either 
by running \PBFT{} among the replicas of the involved shards or 
by starting a two-phase commit protocol after running \PBFT{} locally within each shard, both methods with significant cross-shard costs.

Our \ResilientDB{} fabric enables geo-scale deployment of a fully replicated blockchain system that does not face such challenges, and can easily deal with any workloads. Indeed, the usage of sharding is orthogonal to the fully-replicated design we aim at with \ResilientDB{}. Still, the integration of geo-scale aware sharding with the design of \GeoBFT{}---in which local data is maintained in local clusters only---promises to be an interesting avenue for future work.

\paragraph*{Permissionless Blockchains}
Bitcoin, the first wide-spread permissionless blockchain, uses \emph{Proof-of-Work} (\PoW{}) to replicate data~\cite{encybd,bitcoin,ether}. \PoW{} requires limited communication between replicas, can support many replicas, and can operate in unstructured geo-scale peer-to-peer networks in which independent parties can join and leave at any time~\cite{bitp2p}. Unfortunately, \PoW{} incurs  a high computational complexity on all replicas, which has raised questions about the energy consumption of Bitcoin~\cite{badcoin,badbadcoin}. Additionally, the complexity of \PoW{} causes relative long transaction processing times (minutes to hours) and significantly limits the number of transactions a permissionless blockchain can handle: in 2017, it was reported that Bitcoin can only process 7 transactions per second, whereas Visa already processes 2000 transactions per second on average~\cite{hypereal}. Since the introduction of Bitcoin, several \PoW{}-inspired protocols and Bitcoin-inspired systems have been proposed~\cite{bitcoinng,ouroboros,ppcoin,byzcoin,hybridconsensus,ether}, but none of these proposals come close to providing the performance of traditional permissioned systems, which are already vastly outperformed by \GeoBFT{}.

\section{Conclusions and Future Work}

In this paper, we present our Geo-Scale Byzantine Fault-Tolerant consensus protocol (\GeoBFT{}), a novel consensus protocol with great scalability. To achieve great scalability, \GeoBFT{} relies on a topological-aware clustering of replicas in local clusters to minimize costly global communication, while providing parallelization of consensus. As such, \GeoBFT{} enables geo-scale deployments of high-performance blockchain systems. To support this vision, we implement \GeoBFT{} in our permissioned blockchain fabric---\ResilientDB{}---and show that \GeoBFT{} is not only correct, but also attains up to \emph{six times} higher throughput than existing state-of-the-art \BFT{} protocols.

\balance

\bibliographystyle{plainurl}
\bibliography{reference}

\end{document}